\documentclass[]{article}

\usepackage{fullpage}
\usepackage{amsthm}
\usepackage{microtype}
\usepackage{amssymb}

\usepackage[dvipsnames]{xcolor}

\newtheorem{proposition}[]{Proposition}
\newtheorem{example}[]{Example}
\newtheorem{remark}[]{Remark}
\newtheorem{theorem}[]{Theorem}
\newtheorem{lemma}[]{Lemma}

\newtheorem{definition}[]{Definition}

\usepackage[utf8]{inputenc}

\usepackage{amsmath}
\usepackage{amsthm}

\usepackage{xspace}
\usepackage{booktabs}
\usepackage{colortbl}

\usepackage{tikz}
\usetikzlibrary{decorations.pathreplacing}
\usetikzlibrary{automata}
\tikzset{>=stealth, shorten >=1pt}
\tikzset{every edge/.style = {thick, ->, draw}}
\tikzset{every loop/.style = {thick, ->, draw}}

\newcommand{\myquot}[1]{``#1''}

\newcommand{\nats}{\mathbb{N}}
\newcommand{\set}[1]{\{#1\}}

\newcommand{\size}[1]{|#1|}

\renewcommand{\epsilon}{\varepsilon}
\newcommand{\occ}[2]{|#1|_{#2}}

\newcommand{\NP}[0]{NP\xspace}
\renewcommand{\P}[0]{P\xspace}

\newcommand{\equal}{E}
\newcommand{\nondyck}{N}
\newcommand{\double}{T}

\newcommand{\aut}{\mathcal{A}}

\newcommand{\init}{I}
\newcommand{\run}{\rho}

\newcommand{\runs}[1]{P#1}

\newcommand{\projsigma}{p_\Sigma}
\newcommand{\parikhimage}{\Phi_e}

\newcommand{\pa}{PA\xspace}
\newcommand{\dpa}{DPA\xspace}
\newcommand{\nfa}{NFA\xspace}
\newcommand{\dfa}{DFA\xspace}
\newcommand{\hdpa}{HDPA\xspace}

\newcommand{\uca}{UCA\xspace}
\newcommand{\wupa}{WUPA\xspace}

\newcommand{\mach}{\mathcal{M}}
\newcommand{\stopp}{\texttt{STOP}}
\newcommand{\instr}{\texttt{I}}
\newcommand{\inc}[1]{\texttt{INC(X}_{#1}\texttt{)}}
\newcommand{\dec}[1]{\texttt{DEC(X}_{#1}\texttt{)}}
\newcommand{\ite}[3]{\texttt{IF X}_{#1}\texttt{=0 GOTO }#2\texttt{ ELSE GOTO }#3}

\newcommand{\hstrat}{r}
\newcommand{\hstratiter}{\hstrat^*}

\newcommand{\hstratprime}{t}
\newcommand{\hstratiterprime}{\hstratprime^*}

\newcommand{\hstratprimeinv}[2]{\hstratprime^{-1}_{#1}(#2)}

\newcommand{\attacker}{\textsc{Challenger}\xspace}
\newcommand{\defender}{\textsc{Resolver}\xspace}

\newcommand{\restr}[2]{#1\!\!\upharpoonright\!\! #2}

\newcommand{\attr}[0]{\mathrm{Attr}}
\newcommand{\cpre}[0]{\mathrm{CPre}}

\newcommand{\tsys}{\mathcal{T}}
\newcommand{\traces}{\mathsf{tr}}

\newcommand{\lmark}{\rhd}
\newcommand{\rmark}{\lhd}

\newcommand{\machine}{\mathcal{M}}
\newcommand{\conf}{S}

\newcommand{\rbcm}{RBCM\xspace}
\newcommand{\drbcm}{DRBCM\xspace}
\newcommand{\hdrbcm}{HDRBCM\xspace}

\newcommand{\nrbcm}[1]{{#1}-RBCM\xspace}
\newcommand{\ndrbcm}[1]{{#1}-DRBCM\xspace}
\newcommand{\nhdrbcm}[1]{{#1}-HDRBCM\xspace}
\newcommand{\Sigmalr}{\Sigma_{\lmark\hspace{-.08cm}\rmark}}

\newcommand{\maxpos}{\mathrm{pos}}

\newcommand{\stini}{INI\xspace}
\newcommand{\stinc}{INC\xspace}
\newcommand{\stdec}{DEC\xspace}
\newcommand{\stzero}{ZERO\xspace}

\newcommand{\cycleregister}{\chi}
\newcommand{\red}[1]{\mathrm{red}(#1)}

\newcommand{\sys}{\mathcal{S}}

\begin{document}

\title{History-deterministic Parikh Automata\thanks{Shibashis Guha was supported by SERB grant SRG/2021/000466. Karoliina Lehtinen was supported by ANR QUASY 23-CE48-0008-01. Martin Zimmermann was supported by by DIREC – Digital Research Centre Denmark.}}
\newcommand{\email}[1]{\texttt{#1}}
\date{\vspace{-.3cm}}
\author{\small Enzo Erlich\\
\small Université Paris Cité, CNRS, IRIF, F-75013, Paris, France and EPITA Research Laboratory (LRE), Paris, France \\\medskip
\small\email{Enzo.Erlich@irif.fr}\\
\small Mario Grobler\\
\small University of Bremen, Germany\\\medskip
\small \email{grobler@uni-bremen.de}\\
\small Shibashis Guha\\
\small Tata Institute of Fundamental Research, Mumbai, India\\\medskip 
\small \email{shibashis.guha@tifr.res.in}\\ 
\small Ismaël Jecker\\
\small FEMTO-ST, CNRS, Université de Franche-Comté, Besançon, France\\\medskip
\small \email{ismael.jecker@gmail.com}\\
\small Karoliina Lehtinen\\
\small CNRS, Aix-Marseille University, LIS,
\small Marseille, France\\\medskip
\small \email{lehtinen@lis-lab.fr}\\
\small Martin Zimmermann\\
\small University of Aalborg, Denmark\\
 \small \email{mzi@cs.aau.dk}
}
\maketitle
\begin{abstract}
Parikh automata extend finite automata by counters that can be tested for membership in a semilinear set, but only at the end of a run. Thereby, they preserve many of the desirable properties of finite automata. 
Deterministic Parikh automata are strictly weaker than nondeterministic ones, but enjoy better closure and algorithmic properties.

This state of affairs motivates the study of intermediate forms of nondeterminism.
Here, we investigate history-deterministic Parikh automata, i.e., automata whose nondeterminism can be resolved on the fly. This restricted form of nondeterminism is well-suited for applications which classically call for determinism, e.g., solving games and composition.

We show that history-deterministic Parikh automata are strictly more expressive than deterministic ones, incomparable to unambiguous ones, and enjoy almost all of the closure properties of deterministic automata. 
Finally, we investigate the complexity of resolving nondeterminism in history-deterministic Parikh automata.
\end{abstract}
\pagebreak

\section{Introduction}

Some of the most profound (and challenging) questions of theoretical computer science are concerned with the different properties of deterministic and nondeterministic computation, the \P vs.\ \NP problem being arguably the most important and surely the most well-known one.
However, even in the more modest setting of automata theory, there is a trade off between deterministic and nondeterministic automata with far-reaching consequences for, e.g., the automated verification of finite-state systems.
In the automata-based approach to model checking, for example, one captures a finite-state system~$\sys$ and a specification~$\varphi$ by automata~$\aut_\sys$ and $\aut_\varphi$ and then checks whether $L(\aut_\sys) \subseteq L(\aut_\varphi)$ holds, i.e., whether every execution of $\sys$ satisfies the specification~$\varphi$.
To do so, one tests $L(\aut_\sys) \cap \overline{L(\aut_\varphi)}$ for emptiness. 
Hence, one is interested in expressive automata models that have good closure and algorithmic properties.
Nondeterminism yields succinctness (think \dfa's vs.\ \nfa's) or even more expressiveness (think pushdown automata) while deterministic automata often have better algorithmic properties and better closure properties (again, think, e.g., pushdown automata).

Limited forms of nondeterminism constitute an appealing middle ground as they often combine the best of both worlds, e.g., increased expressiveness in comparison to deterministic automata and better algorithmic and closure properties than nondeterministic ones. 
A classical, and well-studied, example are unambiguous automata, i.e., nondeterministic automata that have at most one accepting run for every input.
For example, unambiguous finite automata can be exponentially smaller than deterministic ones~\cite{leung} while unambiguous pushdown automata are more expressive than deterministic ones~\cite{Hopcroft}.

Another restricted class of nondeterministic automata is that of residual automata~\cite{DLT01}, automata where every state accepts a residual language of the automaton's language. 
For every regular language there exists a residual automaton. While there exist residual automata that can be exponentially smaller than \dfa, there also exist languages for which \nfa can be exponentially smaller than residual automata~\cite{DLT01}.

More recently, another notion of limited nondeterminism has received considerable attention:
history-deterministic automata~\cite{Col09,HP06}\footnote{There is a closely related notion, good-for-gameness, which is often, but not always equivalent~\cite{BL21} (despite frequently being used interchangeably in the past).} are nondeterministic automata whose nondeterminism can be resolved based on the run constructed thus far, but independently of the remainder of the input.
This property makes history-deterministic automata suitable for composition with games, trees, and other automata, applications which classically require deterministic automata.
History-determinism has been studied in the context of regular~\cite{BK18,HP06,KS15}, pushdown~\cite{GJLZ24,LZ22}, quantitative~\cite{DBLP:conf/fossacs/BokerL22,Col09}, and timed automata~\cite{timedhd} (see~\cite{BL23} for a recent survey).
For automata that can be determinized, history-determinism offers the potential for succinctness (e.g., co-Büchi automata~\cite{KS15}) while for automata that cannot be determinized, it even offers the potential for increased expressiveness (e.g., pushdown automata~\cite{GJLZ24,LZ22}). In the quantitative setting, the exact power of history-determinism depends largely on the type of quantitative automata under consideration. So far, it has been studied for quantitative automata in which runs accumulate weights into a value using a value function such as \texttt{Sum, LimInf, Average}, and that assign to a word the supremum among the values of its runs. For these automata, history-determinism turns out to have interesting applications for quantitative synthesis~\cite{BL21}. 
Here, we continue this line of work by investigating history-deterministic Parikh automata. These are a mildly quantitative form of automata with a significant gap between the expressiveness and algorithmic properties of its deterministic and nondeterministic versions, making it a good candidate for studying restricted forms of nondeterminism, such as history-determinism.

Parikh automata, introduced by Klaedtke and Rueß~\cite{KR}, consist of finite automata, augmented with counters that can only be incremented. A Parikh automaton only accepts a word if the final counter-configuration is within a semilinear set specified in the automaton. As the counters do not interfere with the control flow of the automaton, that is, counter values do not affect whether transitions are enabled, they allow for mildly quantitative computations without the full power of vector addition systems or other more powerful models. 

For example the language of words over the alphabet~$\set{0,1}$ having a prefix with strictly more $1$'s than $0$'s is accepted by a Parikh automaton that starts by counting the number of $0$'s and $1$'s and after some prefix nondeterministically stops counting during the processing of the input.
It accepts if the counter counting the $1$'s is, at the end of the run, indeed larger than the counter counting the $0$'s. 
Note that the nondeterministic choice can be made based on the word processed thus far, i.e., as soon as a prefix with more $1$'s than $0$'s is encountered, the counting is stopped.
Hence, the automaton described above is in fact history-deterministic.

Klaedtke and Rueß~\cite{KR} showed Parikh automata to be expressively equivalent to a quantitative version of existential weak MSO that allows for reasoning about set cardinalities. Their expressiveness also coincides with that of reversal-bounded counter machines~\cite{KR}, in which counters can go from decrementing to incrementing only a bounded number of times, but in which counters affect control flow~\cite{Ibarra78RBCM}. The weakly unambiguous restriction of Parikh automata, that is, those that have at most one accepting run, on the other hand, coincide with unambiguous reversal-bounded counter machines~\cite{BCKN20}. Parikh automata are also expressively equivalent to weighted finite automata over the groups~$(\mathbb{Z}^k, +, 0)$~\cite{DM00,MS01} for $k \geqslant 1$. This shows that Parikh automata accept a natural class of quantitative specifications.

Despite their expressiveness, nondeterministic Parikh automata retain some decidability: nonemptiness, in particular, is \NP-complete~\cite{FL} and regular separability is decidable~\cite{CCLP17}. For weakly unambiguous Parikh automata, inclusion is decidable~\cite{CM17}. On the other hand, for deterministic Parikh automata, finiteness, universality, regularity, and model-checking are decidable as well~\cite{CFM12,KR}.
Figueira and Libkin~\cite{FL} also argued that Parikh automata are well-suited for querying graph databases, while mitigating some of the complexity issues related with more expressive query languages. Further, they have been used in the model checking of transducer properties~\cite{FMR20}.
 
As Parikh automata have been established as a robust and useful model, many variants thereof exist: pushdown (visibly~\cite{DFT19} and otherwise~\cite{Kar04}), two-way with~\cite{DFT19} and without stack~\cite{FGM19}, unambiguous~\cite{CFM13}, and weakly unambiguous~\cite{BCKN20} Parikh automata, to name a few. 

\subsection*{Our contributions} 
We introduce history-deterministic Parikh automata~(\hdpa) and study their expressiveness, their closure properties, and their algorithmic properties.

Our main result shows that history-deterministic Parikh automata are more expressive than deterministic ones~(\dpa), but less expressive than nondeterministic ones~(\pa). 
Furthermore, we show that they are of incomparable expressiveness to both classes of unambiguous Parikh automata found in the literature, but equivalent to history-deterministic reversal-bounded counter machines, another class of history-deterministic automata that is studied here for the first time.
These results show that history-deterministic Parikh automata indeed constitute a novel class of  languages capturing quantitative features.

Secondly, we show that history-deterministic Parikh automata satisfy almost the same closure properties as deterministic ones, the only difference being non-closure under complementation.
This result has to be contrasted with unambiguous Parikh automata being closed under complement~\cite[Proposition~5]{CFM13}.
Thus, history-determinism is a too strong form of nondeterminism to preserve closure under complementation, a phenomenon that has already been observed in the case of pushdown automata~\cite{GJLZ24,LZ22}.

Further, we study the algorithmic properties of history-deterministic Parikh automata.
Most importantly, safety model checking for \hdpa is decidable, as it is for \pa. The problem asks, given a system and a set of \emph{bad} prefixes specified by an automaton, whether the system has an execution that has a bad prefix.
This allows, for example, to check properties of an arbiter of some shared resource like “the accumulated waiting time between requests and
responses of client~1 is always at most twice the accumulated waiting time for client~2 and vice
versa”. Note that this property is not $\omega$-regular.

Non-emptiness and finiteness are also both decidable for \hdpa (as they are for nondeterministic automata), but universality, inclusion, equivalence, and regularity are not.
This is in stark contrast to unambiguous Parikh automata (and therefore also deterministic ones), for which \emph{all} of these problems are decidable.
Further, we show that it is undecidable whether a Parikh automaton is history-deterministic and to decide whether it is equivalent to a history-deterministic one.

Finally, we study the complexity of resolving nondeterminism in \hdpa, e.g., the question whether a resolver can be implemented by a \dpa. 
Here, we present a game-theoretic approach, but only obtain a conditional answer.

Note that we consider only automata over finite words here, but many of our results can straightforwardly be transferred to Parikh automata over infinite words, introduced independently by Guha et al.~\cite{GJLZ22} and Grobler et al.~\cite{grobler2023remarks}.

This article is a revised and extended version of an article published at CONCUR 2023~\cite{conference}, which contains all proofs omitted in the conference version and a new section on resolvers implementable by Parikh automata.

\section{Definitions}
\label{sec_definitions}
An alphabet is a finite nonempty set~$\Sigma$ of letters. As usual, $\epsilon$ denotes the empty word, $\Sigma^*$ denotes the set of finite words over $\Sigma$, and $\Sigma^+$ denotes the set of finite nonempty words over $\Sigma$.
The length of a finite word~$w$ is denoted by $\size{w}$ and, for notational convenience, we define $\size{w} = \infty$ for all infinite words~$w$.
Finally, $\occ{w}{a}$ denotes the number of occurrences of the letter~$a$ in a finite word~$w$.

\subsection{Semilinear Sets}

We denote the set of nonnegative integers by $\nats$. 
Given vectors~$\vec{v} = (v_0, \ldots, v_{d-1}) \in \nats^d$ and $\vec{v}\,' = (v'_{0}, \ldots, v'_{d'-1}) \in \nats^{d'}$, we define their concatenation~$\vec{v} \cdot \vec{v}\,' = (v_0, \ldots, v_{d-1}, v'_{0}, \ldots, v'_{d'-1}) \in \nats^{d+d'}$.
We lift the concatenation of vectors to sets $D \subseteq \nats^d$ and $D' \subseteq \nats^{d'}$ via $D\cdot D' = \set{\vec{v} \cdot \vec{v}\,' \mid \vec{v} \in D \text{ and } \vec{v}\,' \in D'}$.

Let $d \geqslant 1$. A set~$C \subseteq \nats^d$ is {linear} if there are vectors~$\vec{v}_0, \ldots, \vec{v}_k \in \nats^d$ such that 
\[
C = \left\{ \vec{v}_0 + \sum\nolimits_{i=1}^k c_i\vec{v}_i \:\middle|\: c_i \in \nats \text{ for } i=1,\ldots, k \right\}.
\]
Furthermore, a subset of $\nats^d$ is {semilinear} if it is a finite union of linear sets. 

\begin{example}
\label{example_sl}
The sets~$\set{(n,n) \mid n\in\nats} =  \set{ (0,0) + c(1,1) \mid c\in\nats }$ and  $\set{(n,2n) \mid n \in\nats} = \set{ (0,0) + c(1,2) \mid c\in\nats }$ are linear, so their union is semilinear. Further, the set~$\set{(n,n') \mid n < n'} = \set{(0,1) + c_1(1,1) + c_2(0,1) \mid c_1, c_2 \in\nats } $ is linear and thus also semilinear.
\end{example}

Ginsburg and Spanier showed that semilinear sets are closed under the Boolean operators~\cite{GS} while closure under concatenation can easily be shown by extending the generating vectors with zeroes.

\begin{proposition}
\label{propsemilinearclosure}
\begin{enumerate}
    \item If $C,C' \subseteq \nats^d$ are semilinear, then so are $C\cup C'$, $C\cap C'$, $\nats^d \setminus C$.
    \item If $C\subseteq \nats^d$ and $C' \subseteq \nats^{d'}$ are semilinear, then so is $C\cdot C'$.
\end{enumerate}
\end{proposition}

In our proofs, it is sometimes convenient to work with semilinear sets defined by formulas of Presburger arithmetic, i.e., first-order formulas over the structure~$\mathcal{N} = (\nats, +, <, 0,1)$.

\begin{proposition}[\cite{GSalgol}]
\label{prop_presburger}
A set~$C \subseteq \nats^d$ is semilinear if and only if there is a formula~$\varphi(x_0, \ldots, x_{d-1})$ of Presburger arithmetic with $d$ free variables such that $C = \set{\vec{v} \in\nats^d \mid \mathcal{N} \models \varphi(\vec{v}) }$.
\end{proposition}

\subsection{Finite Automata}

A (nondeterministic) finite automaton (\nfa)~$\aut = (Q,\Sigma, q_\init, \Delta, F)$ over $\Sigma$ consists of the finite set~$Q$ of states containing the initial state~$q_\init$, the alphabet~$\Sigma$, the transition relation~$\Delta \subseteq Q \times \Sigma \times Q$, and the set~$F \subseteq Q$ of accepting states.
The \nfa is deterministic (i.e.,~a \dfa) if for every state~$q \in Q$ and every letter~$a \in \Sigma$, there is at most one $q'\in Q$ such that $(q,a,q')$ is a transition of $\aut$.

A run of $\aut$ is a (possibly empty) sequence~$(q_0, a_0, q_1) (q_1, a_1, q_2) \cdots (q_{n-1}, a_{n-1}, q_n) $ of transitions with $q_0 = q_\init$. 
It processes the word~$a_0 a_1 \cdots a_{n-1} \in \Sigma^*$, which is the empty word if the run is empty. The run is {accepting} if it is either empty and the initial state is accepting or if it is nonempty and $q_n$ is accepting. 
The language~$L(\aut)$ of $\aut$ contains all finite words $w \in \Sigma^*$ such that $\aut$ has an accepting run processing $w$.

\subsection{Parikh Automata}
Let $\Sigma$ be an alphabet, $d \geqslant 1$, and $D$ a finite subset of $\nats^d$. 
Furthermore, let $w = (a_0,\vec{v}_0) \cdots (a_{n-1},\vec{v}_{n-1})$ be a word over $\Sigma \times D$. 
The $\Sigma$-projection of $w$ is $\projsigma(w) = a_0 \cdots a_{n-1} \in\Sigma^*$ and its \emph{extended Parikh image} is $\parikhimage(w)=\sum_{j=0}^{n-1} \vec{v}_j \in \nats^d$ with the convention~$\parikhimage(\epsilon) =  \vec{0}$, where $\vec{0}$ is the $d$-dimensional zero vector.

A {Parikh automaton} (\pa) is a pair $(\aut, C)$ such that $\aut$ is an \nfa over $\Sigma \times D$ for some input alphabet~$\Sigma$ and some finite $D \subseteq \nats^d$ for some $d \geqslant 1$, and $C \subseteq \nats^d$ is semilinear. 
The language of $(\aut, C)$ consists of the $\Sigma$-projections of words~$w \in L(\aut)$ whose extended Parikh image is in $C$, i.e.,
\[
L(\aut, C) = \set{ \projsigma(w) \mid w \in L(\aut) \text{ with } \parikhimage(w)\in C }.
\]
The Parikh automaton~$(\aut, C)$ is deterministic if for every state~$q$ of $\aut$ and every $a \in \Sigma$, there is at most one pair~$(q', \vec{v}) \in Q \times D$ such that $(q,(a,\vec{v}),q')$ is a transition of $\aut$. Note that this definition does \emph{not} coincide with $\aut$ being a \dfa:
As mentioned above, $\aut$ accepts words over $\Sigma\times D$ while $(\aut, C)$ accepts words over $\Sigma$.
Therefore, determinism is defined with respect to $\Sigma$ only. 
Further, $(\aut, C)$ is complete if for every $q \in Q$ and every $a \in \Sigma$, there is a $\vec{v} \in D$ and a $q' \in Q$ such that $(q,(a,\vec{v}),q') \in \Delta$, i.e., from every state, every letter can be processed. As usual, an incomplete automaton can be completed by adding a fresh non-accepting sink state.

Note that the above definition of $L(\aut, C)$ coincides with the following alternative definition via accepting runs.
\begin{definition}
\label{def_accparuns}
A run $\run$ of $(\aut, C)$ is a run \[\rho = (q_0, (a_0, \vec{v}_0), q_1) (q_1, (a_1, \vec{v}_1), q_2) \cdots (q_{n-1}, (a_{n-1}, \vec{v}_{n-1}), q_{n})\] of $\aut$. 
We say that $\rho$ \emph{processes} the word~$a_0 a_1 \cdots a_{n-1}\in \Sigma^*$, i.e., the $\vec{v}_j$ are ignored, and that $\rho$'s extended Parikh image is $\sum_{j=0}^{n-1} \vec{v}_j$.
The run is accepting if it is either empty and both the initial state of $\aut$ is accepting and the zero vector (the extended Parikh image of the empty run) is in $C$, or if it is nonempty, $q_n$ is accepting, and $\rho$'s extended Parikh image is in $C$.
Finally, $(\aut, C)$ accepts $w \in \Sigma^*$ if it has an accepting run processing~$w$.
\end{definition}

\begin{example}\label{example_intro}
Consider the deterministic \pa~$(\aut, C)$ with $\aut$ in Figure~\ref{fig_example_intro} and $
C = \set{(n,n) \mid n\in\nats} \cup \set{(n,2n) \mid n \in\nats}$ (cf.\ Example~\ref{example_sl}). 
It accepts the language~$\set{a^nb^n \mid n\in\nats} \cup \set{a^nb^{2n} \mid n \in\nats}$.
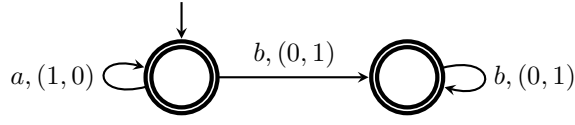
\begin{figure}[h]
    \centering
    \begin{tikzpicture}[ultra thick]
    
    \node[state,accepting] (1) at (0,0) {};
    \node[state,accepting] (2) at (3,0) {};
    
    \path[-stealth]
    (0,1) edge (1)
    (1) edge[loop left] node[left] {$a,(1,0)$} ()
    (1) edge node[above] {$b,(0,1)$} (2)
    (2) edge[loop right] node[right] {$b,(0,1)$} ();
    
    \end{tikzpicture}
    \caption{The automaton for Example~\ref{example_intro}.}
    \label{fig_example_intro}
\end{figure}
\end{example}

Note that there is a closely related automaton model, called constrained automata~\cite{CFM12}, that has the same expressive power as \pa as introduce here~\cite[Theorem~3.4.]{CFM12}.
Here, we mostly work with \pa, but we compare unambiguous constrained automata in Subsection~\ref{subsec_unamb} to history-deterministic \pa.

\section{History-deterministic Parikh Automata}
\label{sec_hd}

In this section, we introduce history-deterministic Parikh automata and give examples. 

Let $(\aut, C)$ be a \pa with $\aut = (Q, \Sigma \times D, q_\init, \Delta, F)$.
For a function~$\hstrat\colon \Sigma^+ \rightarrow \Delta$ we define its iteration $\hstratiter \colon \Sigma^* \rightarrow \Delta^*$ via $\hstratiter(\epsilon) = \epsilon$ and $\hstratiter(a_0 \cdots a_n) = \hstratiter(a_0 \cdots a_{n-1})\cdot \hstrat(a_0 \cdots a_n)$. 
We say that $\hstrat$ is a resolver for $(\aut, C)$ if, for every~$w \in L(\aut, C)$, $\hstratiter(w)$ is an accepting run of $(\aut, C)$ processing $w$. 
Further, we say that $(\aut, C)$ is history-deterministic (i.e., an \hdpa) if it has a resolver.

\begin{example}
\label{example_nondyck}
Fix $\Sigma = \set{0,1}$ and say that a word~$w \in \Sigma^*$ is non-Dyck if $\occ{w}{0} < \occ{w}{1}$. 
We consider the language $\nondyck \subseteq \Sigma^+$ of words that have a non-Dyck prefix. 
It is accepted by the \pa~$(\aut, C)$ where $\aut$ is depicted in Figure~\ref{fig_sep} and $C = \set{(n,n') \mid n<n'}$ (cf.\ Example~\ref{example_sl}).
Intuitively, in the initial state~$q_c$, the automaton counts the number of $0$'s and $1$'s occurring in some prefix, nondeterministically decides to stop counting by moving to $q_n$ (this is the only nondeterminism in $\aut$), and accepts if there are more $1$'s than $0$'s in the prefix.

\begin{figure}[h]
    \centering
    \begin{tikzpicture}[ultra thick]
    
    \node[state] (1) at (0,0) {$q_c$};
    \node[state,accepting] (2) at (4,0) {$q_n$};
    
    \path[-stealth]
    (0,1) edge (1)
    (1) edge[loop left] node[left,align=left] {$0,(1,0)$\\$1,(0,1)$} ()
    (1) edge node[above,align=left] {$0,(1,0)$\\$1,(0,1)$} (2)
    (2) edge[loop right] node[right,align=left] {$0,(0,0)$\\$1,(0,0)$} ();
    
    \end{tikzpicture}
    \caption{The automaton for Example~\ref{example_nondyck}.}
    \label{fig_sep}
\end{figure}
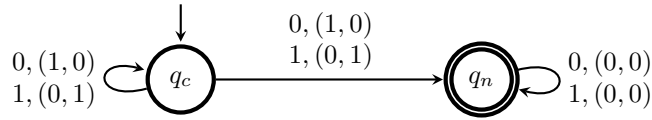

The nondeterministic choice can be made only based on the prefix processed thus far, i.e., as soon as the first non-Dyck prefix is encountered, the resolver proceeds to state~$q_n$, thereby ending the prefix. 
Formally, the function
\[
wb \mapsto \begin{cases}
(q_c, (b,(1-b,b)),q_c) &\text{if $wb$ has no non-Dyck prefix,}\\
(q_c, (b,(1-b,b)),q_n) &\text{if $wb$ is non-Dyck, but $w$ has no non-Dyck prefix,}\\
(q_n, (b,(0,0)),q_n) &\text{if $w$ has a non-Dyck prefix,}\\
\end{cases}
\]
is a resolver for $(\aut, C)$.
\end{example}

\begin{remark}
As a resolver resolves nondeterminism and a \dpa has no nondeterminism to resolve, every \dpa is history-deterministic.
\end{remark}

\section{Expressiveness}
\label{sec_expressiveness}

In this section, we study the expressiveness of \hdpa by comparing them to related automata models, e.g., deterministic and nondeterministic Parikh automata, unambiguous Parikh automata (capturing another restricted notion of nondeterminism), and reversal-bounded counter machines (which are known to be related to Parikh automata).
Overall, we obtain the relations shown in Figure~\ref{fig_classes},
where the additional classes of languages and the separating languages will be introduced throughout this section. 

\begin{figure}[h]
    \centering
    
    \begin{tikzpicture}[ultra thick,xscale=1.5,yscale=1.2]
    
\draw[rounded corners,fill=black!5] (-3.5,-2.2) rectangle (1.25,2);

\draw[rounded corners] (-4,-2.5) rectangle (5.5,2.3);
\node at (3.4,-2) {\pa/\rbcm/\hdrbcm};

\draw[rounded corners,fill=black!10] (-1.5,-1.5) rectangle (5,1.5);
\node at (4,-1) {\wupa};

\draw[rounded corners,fill=black!20] (-1.25,-1.2) rectangle (3,.8);
\node at (2.2,-.75) {\uca};

\draw[rounded corners,fill=black!30] (-1,-.5) rectangle (1,.5);
\node at (0,0) {\dpa};

\draw[rounded corners] (-3.5,-2.2) rectangle (1.25,2);
\node at (-1.8,-1.9) {\hdpa/\nhdrbcm{1}};

\draw[fill] (1.5,0) circle (.05cm);
\node[anchor =west] at (1.5,0) {$\equal$ (see Page~\pageref{page_langdef_equal})};
\draw[fill] (3.6,0) circle (.05cm);
\node[anchor =west] at (3.6,0) {$\nondyck' \cup \equal$};
\draw[fill] (-0.5,-.85) circle (.05cm);
\node[anchor =west] at (-0.5,-.85) {$\equal'$ (see Page~\pageref{page_langdef_equal'})};
\draw[fill] (-0.5,1.15) circle (.05cm);
\node[anchor =west] at (-0.5,1.15) {$\nondyck'$ (see Page~\pageref{page_langdef_nondyck'})};
\draw[fill] (-3,-.25) circle (.05cm);
\node[anchor =west] at (-3,-.25) {$\double$ (see Page~\pageref{page_langdef_double})};
\draw[fill] (3,1.85) circle (.05cm);
\node[anchor =west] at (3,1.85) {$\double \cup \equal$};

    \end{tikzpicture}
    \caption{The classes of languages accepted by different models of Parikh automata.}
    \label{fig_classes}
\end{figure}
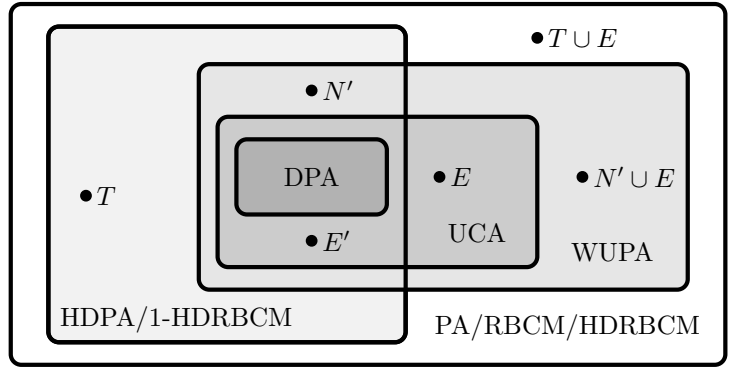

We begin by stating and proving a pumping lemma for \hdpa.
The basic property used here, just as for the pumping lemmata for \pa and \dpa~\cite[Lemma 3.5 and Lemma 3.6]{CFM12}, is that shuffling around cycles of a run does not change whether it is accepting or not, as acceptance only depends on the last state of the run being accepting and the vectors (and their multiplicity) that appear on the run, but not the order of their appearance.

\begin{lemma}
\label{lemma_pumping}
Let $(\aut, C)$ be an \hdpa with $L(\aut, C) \subseteq \Sigma^*$. 
Then, there exist $p, \ell \in \nats$ such that every $w \in \Sigma^*$ with $\size{w}> \ell$ can be written as $w = u v x v z$ such that
\begin{itemize}
    \item $0 < \size{v} \leqslant p$, $\size{x} > p$, and $\size{uvxv} \leqslant \ell$, and
    \item for all $z' \in \Sigma^*$: if $uvxvz' \in L(\aut, C)$, then also $uv^2xz' \in L(\aut, C)$ and $uxv^2z' \in L(\aut, C)$.
\end{itemize}
\end{lemma}

\begin{proof}
Fix some resolver~$\hstrat$ for $(\aut,C)$. 
Note that the definition of a resolver only requires $\hstratiter(w)$ to be a run processing~$w$ for those $w \in L(\aut, C)$.
Here, we assume without loss of generality that $\hstratiter(w)$ is a run processing $w$ for each $w \in \Sigma^*$. 
This can be achieved by completing $\aut$ (by adding a nonaccepting sink state and transitions to the sink where necessary) and redefining $\hstrat$ where necessary (which is only the case for inputs that cannot be extended to a word in $L(\aut, C)$).

A cycle is a nonempty finite run infix
\[(q_0, a_0, q_1) (q_1, a_1, q_2) \cdots (q_{n-1}, a_{n-1}, q_n)(q_n, a_n, q_0)\]
starting and ending in the same state and such that the $q_j$ for $j \in\set{0,1,\ldots, n}$ are all pairwise different. 
Now, let $p$ be the number of states of $\aut$ and let $m$ be the number of cycles of $\aut$.
Note that every run infix containing at least $p$ transitions contains a cycle.

We define $\ell = (p+1) (2m+1)$, consider a word~$w \in \Sigma^*$ with $\size{w} > \ell$, and let $\rho = \hstratiter(w) $
be the run of $\aut$ induced by $\hstrat$, which processes $w$.
We split $\rho$ into $\rho_0 \rho_1 \cdots \rho_{2m}\rho'$ such that each $\rho_j$ contains $p+1$ transitions.
Then, each $\rho_j$ contains a cycle and there are $j_0,j_1$ with $j_1 > j_0+1$ such that $\rho_{j_0}$ and $\rho_{j_1}$ contain the same cycle. 
Now, let 
\begin{itemize}
    \item $\rho_v$ be the cycle in $\rho_{j_0}$ and $\rho_{j_1}$,
    \item $\rho_u$ be the prefix of $\rho$ before the first occurrence of $\rho_v$ in $\rho_{j_0}$, and
    \item $\rho_x$ be the infix of $\rho$ between the first occurrences of $\rho_v$ in $\rho_{j_0}$ and $\rho_{j_1}$.
\end{itemize}
Furthermore, let $u, v,x \in \Sigma^*$ be the inputs processed by $\rho_u$, $\rho_v$, and $\rho_x$ respectively.
Then, we indeed have $0 < \size{v} \leqslant p$ (as we consider simple cycles), $\size{x} > p$ (as $j_1 > j_0+1$), and $\size{uvxv} \leqslant \ell$.

Note that $\rho_u\rho_v\rho_x\rho_v$, $\rho_u \rho_v^2\rho_x$, and $\rho_u \rho_x\rho_v^2$ are all runs of $\aut$ which process $uvxv$, $uv^2x$, and $uxv^2$ respectively. 
Furthermore, all three runs end in the same state and their extended Parikh images are equal, as we only shuffled pieces around.

Now, consider some $z'$ such that $uvxvz' \in L(\aut, C)$. Then $\hstratiter(uvxvz')$ is an accepting run, and of the form $\rho_u \rho_v \rho_x \rho_v \rho_{z'}$ for some $\rho_{z'}$ processing $z'$.
Now, $\rho_u \rho_v^2\rho_x \rho_{z'}$, and $\rho_u \rho_x\rho_v^2\rho_{z'}$ are accepting runs of $(\aut, C)$ (although not necessarily induced by $\hstrat$) processing $uv^2xz'$ and $uxv^2z'$, respectively. 
Thus, $uv^2xz' \in L(\aut, C)$ and $uxv^2z' \in L(\aut, C)$.
\end{proof}

It is instructive to compare our pumping lemma for \hdpa to those for \pa and \dpa~\cite[Lemma 3.5 and Lemma 3.6]{CFM12}:
\begin{itemize}
    \item The pumping lemma for \pa states that every \emph{long} word accepted by a \pa can be decomposed into $uvxvz$ as above such that both $uv^2xz$ and $uxv^2z$ are accepted as well. 
    This statement is weaker than ours, as it only applies to the two words obtained by moving a $v$ while our pumping lemma applies to any suffix $z'$.
    This is possible, as the runs of an \hdpa on words of the form~$uvxvz'$ (for fixed $uvxv$), induced by a resolver, all coincide on their prefixes processing $uvxv$.
    This is not necessarily the case in \pa.
    
    \item The pumping lemma for \dpa states that every \emph{long} word (not necessarily accepted by the automaton) can be decomposed into $uvxvz$ as above such that $uvxv$, $uv^2x$, and $uxv^2$ are all equivalent with respect to the Myhill-Nerode equivalence. 
    This statement is stronger than ours, as Myhill-Nerode equivalence is concerned both with the language of the automaton and its complement. But similarly to our pumping lemma, the one for \dpa applies to all possible suffixes~$z'$. 
\end{itemize}

Now, we apply the pumping lemma to compare the expressiveness of \hdpa, \dpa, and \pa.

\begin{theorem}
\label{thm_sep}
\hdpa are more expressive than \dpa, but less expressive than \pa.
\end{theorem}

\begin{proof}
First, we consider the separation between \dpa and \hdpa.
The language~$\nondyck$ from Example~\ref{example_nondyck}, which is accepted by an \hdpa, is known to be not accepted by any \dpa: 
\dpa are closed under complementation~\cite[Property~4]{KR} while the complement of $\nondyck$ is not even accepted by any \pa~\cite[Proposition~11]{CFM13}.

To show that \pa are more expressive than \hdpa, consider the language~\[\label{page_langdef_equal}\equal = \set{a,b}^* \cdot \set{a^n b^n \mid n > 0},\] which can easily be seen to be accepted by a \pa.
We show that $\equal$ is not accepted by any \hdpa\footnote{Note that the related language~$\set{a,b}^*\cdot \set{a^n\# a^n \mid n\in\nats}$ is not accepted by any \dpa~\cite[Proposition~3.8]{CFM12}.}  via an application of the pumping lemma.

% To this end, assume there is some \hdpa~$(\aut, C)$ accepting $\equal$, and let $p, \ell$ as in the pumping lemma.
% We pick~$w = (a^pb)^\ell$, which we decompose as $uvxvz$ with the properties guaranteed by the pumping lemma.
% In particular, we have $0 < \size{v} \leqslant p$ and $\size{x} > p$, which implies that $x$ contains at least one $b$ and $v$ at most one $b$.
% Finally, let $k \geqslant 0$ be the number of $a$'s at the end of $uvxv$. We consider two cases.

% First, assume we have $v = a^j$ for some $1 \leqslant j \leqslant p$, which implies $k \geqslant j \geqslant 1$. 
% Then, $uv^2x$ ends in $k-j < k$ occurrences of $a$. 
% As $(uvxv)ab^{k+1} \in \equal = L(\aut, C)$, the pumping lemma implies $(uv^2x)ab^{k+1} \in L(\aut, C)$. 
% However, this yields a contradiction as $(uv^2x)ab^{k+1}$ ends in $ba^{(k-j)+1}b^{k+1}$ with $(k-j)+1 < k + 1$ and is therefore not in $\equal$.

% Otherwise, we have $v = a^jba^k$ (recall that $v$ contains at most one $b$), which satisfies $j+k < p$. 
% In this case, $uv^2x$ ends in $p-j > k$ occurrences of $a$. 
% Again, $(uvxv)ab^{k+1} \in \equal = L(\aut, C)$ and the pumping lemma implies $(uv^2x)ab^{k+1} \in L(\aut, C)$, but it ends in $ba^{(p-j)+1}b^{k+1}$ and is therefore not in $\equal$.

To this end, assume there is some \hdpa~$(\aut, C)$ accepting $\equal$, and let $p, \ell$ as in the pumping lemma.
We pick~$w = (a^{p+1}b^{p+1})^\ell$,
which we decompose as $uvxvz$ with the properties guaranteed by the pumping lemma.
In particular, we have $\size{v} \leqslant p$,
therefore $v \in a^* b^* + b^* a^*$.
We consider two cases depending on the last letter of $v$.
In each one, we
show the existence of a word $z'$ such that the word
$uvxvz'$ is in the language $\equal$,
yet either $uv^2xz'$ or $uxv^2z'$ is not.
This yields the desired contradiction to the pumping lemma.

\begin{enumerate}
    \item
    First, assume that the last letter of $v$ is an $a$.
    Since $\size{x} > p$ and $x$ appears between two copies of $v$ in $(a^{p+1}b^{p+1})^\ell$,
    the infix~$xv$ contains at least one
    full $b$-block: we have $xv = x' b^{p+1}a^k$
    with $x' \in \{a,b\}^*$ and $0 < k \leqslant p+1$.
    We set $z' = a^{p+1-k}b^{p+1}$.
    Hence, $uvxvz' = uv x' b^{p+1}a^{p+1}b^{p+1} \in \equal$.
    We show that $uv^2xz' \not\in \equal$
    by differentiating two cases:
    \begin{enumerate}
        \item 
        If $v = a^i$ for some $i$, which must satisfy $0 < i \leqslant k$,
        then $uv^2xz'$ is not in $\equal$
        as it ends with $b^{p+1}a^{p+1-i}b^{p+1}$.
        \item
        Otherwise, we must have $v = b^ia^k$ with $0 < i < p$.
        Then, $uv^2xz'$ is not in $\equal$
        as it ends with $b^{p+1-i}a^{p+1-k}b^{p+1}$.
    \end{enumerate}
    \item
    Otherwise, the last letter of $v$ is a $b$.
    Since $\size{x} > p$ and $x$ appears between two copies of $v$ in $(a^{p+1}b^{p+1})^\ell$,
    the infix~$xv$ contains at least one
    full $a$-block: we have $xv = x'a^{p+1}b^k$ with $x' \in \{a,b\}^*$ and $0 < k \leqslant p+1$.
    This time we set $z' = b^{p+1-k}$.
    Thus, $uvxvz' = uvx' a^{p+1}b^{p+1} \in \equal$,
    and we differentiate two cases to show that
    $uxv^2z' \notin \equal$:
    \begin{enumerate}
        \item 
        If $v = b^i$ for some $i$, which must satisfy $0 < i \leqslant p$, then $uxv^2z'$ ends with $b^{p+i+1}$.
        However, each of its $a$-blocks has length $p+1$, as moving $v = b^i$ with $i \leqslant p$ does not merge any $a$-blocks. Hence, $uxv^2z'$ is not in $\equal$.
        \item
        Otherwise, we must have $v = a^ib^k$ with $0 < i < p$. Then, $uxv^2z'$ is not in $\equal$
        as it ends with $v^2z' = a^ib^ka^ib^{p+1}$.\qedhere
    \end{enumerate}
\end{enumerate}
\end{proof}

\subsection{History-determinism vs.\ Unambiguity}
\label{subsec_unamb}
After having placed history-deterministic Parikh automata strictly between deterministic and nondeterministic ones, we now compare them to unambiguous Parikh automata, another class of automata whose expressiveness lies strictly between that of \dpa and \pa. 
In the literature, there are two (nonequivalent) forms of unambiguous Parikh automata. We consider both of them here.

Cadilhac et al.\ studied unambiguity in Parikh automata in the guise of unambiguous constrained automata (\uca)~\cite{CFM13} (recall that (possibly ambiguous) constrained automata are effectively equivalent to \pa~\cite[Theorem~3.4]{CFM12}). 
Intuitively, an \uca~$(\aut, C)$ over an alphabet~$\Sigma$ consists of an unambiguous $\epsilon$-\nfa~$\aut$ over $\Sigma$, say with $d$ transitions, and a semilinear set~$C \subseteq \nats^d$, i.e., $C$ has one dimension for each transition in $\aut$.
It accepts a word~$w \in \Sigma^*$ if $\aut$ has an accepting run processing~$w$ (due to unambiguity this run must be unique) such that the Parikh image of the run (recording the number of times each transition occurs in the run) is in $C$.

On the other hand, Bostan et al.\ introduced so-called weakly-unambiguous Parikh automata (\wupa)~\cite{BCKN20}.
Intuitively, a \wupa~$(\aut,C)$ over $\Sigma$ is a classical \pa as introduced here where every input over $\Sigma$ has at most one accepting run processing it (in the sense of Definition~\ref{def_accparuns}), i.e., the vectors labelling the transitions are ignored as described in that definition.
Bostan et al.\ discuss the different definitions of unambiguity and in particular show that every \uca is a \wupa, but that \wupa are strictly more expressive (see Remark~11 of \cite{BCKN20} and the references therein). 
Here, we compare the expressiveness of \hdpa to that of \uca and \wupa.

\begin{theorem}
\label{theorem_hdpa_vs_unamb}
The expressiveness of \hdpa is neither comparable with that of \uca nor with that of \wupa.
\end{theorem}

\begin{proof}
The language~$\equal$ from the proof of Theorem~\ref{thm_sep} is accepted by an \uca~\cite[Theorem~12]{CFM13} and thus also by a \wupa, as every \uca can be turned into an equivalent \wupa (again, see Remark~11 of \cite{BCKN20}).
However, $\equal$ is not accepted by any \hdpa, as shown in the proof of Theorem~\ref{thm_sep}.
This yields the first two separations.

Conversely, consider the language\label{page_langdef_double}
\[
\double = \set{ c^{n_0} d c^{n_1} d \cdots c^{n_k} d \mid k \geqslant 1, n_0 = 1 \text{, and } n_{j+1} \neq 2n_j \text{ for some } 0 \leqslant j < k}.
\]
Baston et al.\ proved that $\double$ is not accepted by any \wupa, and therefore also not by any \uca.
We show that it is accepted by an \hdpa, yielding the other two separations.

We start by giving some intuition.
Let $w = c^{n_0} d c^{n_1} d \cdots c^{n_k} d\in \double$ and let $j'\in\nats$ be minimal with $n_{j'+1} \neq 2n_{j'}$.
Then, we have $n_{j} = 2^{j}$ for all $j \leqslant j'$.
Our automaton sums up the $n_{j}$ for even and odd $j$ and relies on some basic facts about sums of powers of two to accept the language.

Consider, for example, the sums of the $2^j$ for even and odd $j \leqslant 5$, respectively. 
The former is $e_{\leqslant 5} = 2^0 + 2^2 + 2^4 = 21$ and the latter is $o_{\leqslant5} = 2^1 + 2^3 + 2^5 = 42$, 
We have $2\cdot e_{\leqslant 5} = o_{\leqslant 5}$ as the terms of the second sum are obtained by doubling the terms of the former one.
Similarly, we have $e_{\leqslant 6} = 85 = 2\cdot 42 + 1 = 2 \cdot o_{\leqslant 6}+1$.
Obviously, these equations hold for arbitrary bounds, i.e., if $j$ is odd then we have $2\cdot e_{\leqslant j} = o_{\leqslant j}$ and if $j$ is even then we have $e_{\leqslant j} = 2 \cdot o_{\leqslant j}+1$. 

Recall that $j'$ was chosen minimally with $n_{j'+1} \neq 2n_{j'}$.
So, the equations described above hold for all $ j \leqslant j'$, but they fail to hold for $j = j'+1$.\footnote{Note that the equation might be satisfied again, e.g., if the \emph{error} $n_{j'+1} - 2n_{j'}$ is compensated by $n_{j'+2}$.
However, this is irrelevant for our argument.}

Figure~\ref{fig_double-aut} depicts an \hdpa~$\aut$ that we show to accept $\double$. 
Intuitively, it sums up the $n_j$ for even and odd $j$ and nondeterministically decides to stop the summation at the end of one such block.
In addition to summing up the $n_j$, $\aut$ also keeps track of $j$ by counting the $d$'s.
Thus, we equip it with the semilinear set
\[
C = \set{(e,o,j) \in \nats^3 \mid 2\cdot e \neq o \text{ and } j \text{ is odd}} \cup \set{(e,o,j) \in \nats^3 \mid  e \neq 2\cdot o+1 \text{ and } j \text{ is even}},
\]
i.e., we check that the above equations were violated.

\begin{figure}[h]
    \centering
    \begin{tikzpicture}[ultra thick]
    
    \node[state] (-1) at (-2.5,0.5) {};
    \node[state] (0) at (-2.5,2.5) {};
    \node[state] (1) at (0,2.5) {};
    \node[state] (2) at (0,0.5) {};
    \node[state,accepting] (3) at (5,1.5) {};
    \node[state] (4) at (7.5,1.5) {};
    
    \path[-stealth]
    (-3.5,.5) edge (-1)
    (-1) edge[bend left] node[left] {$c,(1,0,0)$} (0)
    (0) edge[bend left] node[above] {$d,(0,0,1)$} (1)
    (1) edge[bend left] node[right] {$d,(0,0,1)$} (2)
    (2) edge[bend left] node[left] {$d,(0,0,1)$} (1)
    (1) edge[loop above] node[above] {$c,(0,1,0)$} ()
    (2) edge[loop below] node[below] {$c,(1,0,0)$} ()
    (1) edge[bend left] node[above] {$d,(0,0,1)$} (3)
    (2) edge[bend right] node[below] {$d,(0,0,1)$} (3)
    (3) edge[loop left] node[left] {$d,(0,0,0)$} (3)
    (4) edge[loop right] node[right] {$c,(0,0,0)$} ()
    (3) edge[bend left] node[above] {$c,(0,0,0)$} (4)
    (4) edge[bend left] node[below] {$d,(0,0,0)$} (3)
    ;
    
    \end{tikzpicture}
    \caption{The automaton for the language $\double$.}
    \label{fig_double-aut}
\end{figure}

First, let us argue that $(\aut,C)$ accepts $\double$. 
If $w$ is in $\double$, then there is an accepting run processing $w$ that moves to the accepting state as soon as the $(j'+1)$-th $d$ is processed, where $j'$ is defined as above. 
Whether this is the case only depends on the prefix ending at that position, i.e., the nondeterminism can be resolved by a resolver.
Finally, if $w$ is not in $\double$, then there are three cases.
Either, $w$ does not start with $cd$, $w$ does not end with a $d$, or it is of the form~$c^1dc^2dc^4d\cdots c^{2^j}d$ for some $j \geqslant1$. 
In the first two cases, there is no accepting run of the underlying automaton~$\aut$ that processes $w$. 
In the last case, the equations described above are never violated when processing a $d$, so whenever a run ends in the accepting state, the extended Parikh image of the run is not in $C$.
So, $(\aut, C)$ does indeed accept $\double$ and we have argued above that the only nondeterministic choice during any run can be made by a resolver, i.e., $(\aut, C)$ is an \hdpa.
\end{proof}

Finally, we show that all intersections between the different classes introduced above are nonempty.

\begin{theorem}
\label{theorem_hdpa_vs_unamb_fine}
\hfill
\begin{enumerate}
    \item\label{intitemnotdpa} There is a language that is accepted by an \hdpa and by an \uca, but not by any \dpa.
    
    \item\label{intitemnotuca} There is a language that is accepted by an \hdpa and by a \wupa, but not by any \uca.
    
    \item\label{intitem3} There is a language that is accepted by a \pa, but not by any \hdpa nor by any \wupa.
    
    \item\label{intitem4} There is a language that is accepted by a \wupa, but not by any \hdpa nor by any \uca.
\end{enumerate}
\end{theorem}

\begin{proof}
\ref{intitemnotdpa}.) Consider the language\label{page_langdef_equal'}
\[
\equal' = \set{c^m\set{a,b}^{m-1} b a^nb^n \mid m,n>0} 
\]
and compare it to the language~$\equal = \set{a,b}^* \cdot \set{a^nb^n \mid n>0} $ from Theorem~\ref{thm_sep}, which is not accepted by any \hdpa and thus also not by any \dpa.
The intuitive reason is that such an automaton has to guess when the suffix of the form~$a^nb^n$ starts, which cannot be done by a resolver.
However, by adding the $c$'s, which encode the length of the infix before the suffix of the form~$a^nb^n$ starts, a resolver can determine when the suffix starts.
Note that we also, for reasons that we discuss below, require that the last letter before the suffix of the form~$a^nb^n$ is a $b$.

\begin{figure}[h]
    \centering
    \begin{tikzpicture}[ultra thick]
    
    \node[state] (1) at (0,0) {};
    \node[state] (2) at (3,0) {};
    \node[state] (3) at (6,0) {$q$};
    \node[state] (4) at (9,0) {};
    \node[state,accepting] (5) at (12,0) {};
    
    \path[-stealth]
    (-1,0) edge (1)
    (1) edge[loop above] node[above,align=left] {$c,(1,0,0,0)$} ()
    (1) edge node[above,align=left] {$a,(0,1,0,0)$\\$b,(0,1,0,0)$} (2)
    (2) edge[loop above] node[above,align=left]
    {$a,(0,1,0,0)$\\$b,(0,1,0,0)$} ()
    (2) edge node[above,align=left] {$b,(0,1,0,0)$} (3)
    (3) edge node[above,align=left] {$a,(0,0,1,0)$} (4)
    (4) edge[loop above] node[above,align=left]
    {$a,(0,0,1,0)$} ()
    (4) edge node[above,align=left] {$b,(0,0,0,1)$} (5)
    (5) edge[loop above] node[above,align=left]
    {$b,(0,0,0,1)$} ()
    (1) edge[bend right=30] node[xshift=0cm, below,align=left] {$b,(0,1,0,0)$} (3)
    % (1) edge[bend right=35] node[above, near end,align=left,xshift=-.6cm] {$a,(0,0,1,0)$} (4)
    
        % (2) edge[loop above] node[above,align=left] {$0,(0,0)$\\$1,(0,0)$} ()
    ;
    
    \end{tikzpicture}
    \caption{The automaton for the language $\equal'$.}
    \label{fig_aut-equalprime}
\end{figure}

The automaton~$(\aut,C)$ with $\aut$ in Figure~\ref{fig_aut-equalprime} and \[C=
\set{(m,m,n,n)\mid m,n \in \nats}
\]
is an \hdpa accepting $\equal'$. 
The automaton accepts~$\equal'$ (when viewed as a \pa) and the only nondeterministic choice, i.e., when to move to $q$ can be made based on the prefix processed thus far, i.e., $q$ has to be reached with the (unique) prefix of the form~$c^m\set{a,b}^{m-1}b$.

Furthermore, consider the \nfa~$\aut'$ over $\set{a,b,c}$ obtained from $\aut$ in Figure~\ref{fig_aut-equalprime} by projecting away the vectors on the transitions, which is unambiguous:
Every word accepted by $\aut'$ must end with $ba^+b^+$ and the first $b$ of that suffix has to lead to $q$. 
As the only nondeterminism in the automaton is the choice to go to $q$ or not, this implies that $\aut'$ has indeed at most one accepting run on every word.

Now, as every vector projected away from $\aut$ is a unit vector, it is straightforward to give an eleven-dimensional\footnote{Note that the automaton in Figure~\ref{fig_aut-equalprime} has eleven transitions.} semilinear set~$C'$ such that $(\aut', C')$ is an \uca accepting~$\equal'$, i.e., $C'$ \emph{simulates} $C$: $C'$ ensures that the initial $c$-labeled self-loop occurs in a run as often as transitions labeled by $(0,1,0,0)$ in Figure~\ref{fig_aut-equalprime} occur in the run (this simulates the first two components of vectors in $C$ being equal). 
Similarly, $C'$ ensures that the transitions labeled by $(0,0,1,0)$ occur as often as transitions labeled by $(0,0,0,1)$. 

Note that requiring a $b$ in front of the suffix~$a^nb^n$ allows us to use the same \emph{transition structure} for $\aut$ and $\aut'$, as both automata make the same nondeterministic choice, i.e., they guess when the suffix of the form~$a^nb^n$ starts. 
Thus, the separation between languages accepted by \dpa and languages accepted both by \hdpa and \uca can be witnessed by a language where the \hdpa and the \uca are essentially the same automaton.
Also, they both rely on the same nondeterministic guess, which can be made history-deterministically and unambiguously.
Without the $b$, the unambiguous automaton would have to guess the start of the longest suffix of the form~$a^nb^{n'}$ with $n \geqslant n'$, and thus an \uca accepting $\equal$ would require a slightly different transition structure than the one shown in Figure~\ref{fig_aut-equalprime}.

Finally, it remains to argue that $\equal'$ is not accepted by any \dpa.
First, let us remark that the pumping argument in the proof of Theorem~\ref{thm_sep} also shows that the language~$\equal_b = \set{a,b}^*\cdot b\cdot \set{ a^n b^n \mid n > 0}$ is not accepted by any \hdpa and thus also not by any \dpa.
Now, we show that a \dpa accepting $\equal'$ can be turned into a \dpa accepting $\equal_b$, which yields the desired result.

So, assume there is a \dpa~$(\aut, C)$ accepting $\equal'$, say with $\aut = (Q, \set{a,b,c} \times D, q_\init, \Delta, F)$.
For every $q \in Q$ let $R_q$ be the set of runs of $\aut$ starting in $q_\init$, processing a word in $c^+$, and ending in $q$.
Furthermore, let
\[
C_q = \left\{ \sum\nolimits_{j=0}^{n-1}\vec{v}_j \mid (q_\init, (c, \vec{v}_0),q_1)(q_1, (c, \vec{v}_1),q_2) \cdots (q_{n-1}, (c, \vec{v}_{n-1}),q) \in R_q \right\}  
\]
be the set of extended Parikh images of those runs, which is semilinear~\cite{Par66}.

Furthermore, for each $q \in Q$ with nonempty $R_q$ let $\aut_q = (Q, \set{a,b}\times D, q, \Delta', F)$ where $\Delta'$ is obtained by removing all $c$-transitions from $\Delta$. 
Note that each $\aut_q$ is still deterministic, as we have only changed the initial state and removed transitions.
Finally, we define 
\[C'_q = \set{\vec{v} \mid \text{there exists }\vec{v}\,' \in C_q \text{ such that } \vec{v} + \vec{v}\,' \in C},
\]
which is again semilinear, as it can be defined by a Presburger formula constructed from Presburger formulas for $C_q$ and $C$ (see Proposition~\ref{prop_presburger}).

We claim $\equal_b = \bigcup_q L(\aut_q, C_q')$ where $q$ ranges over all states such that $R_q$ is nonempty. 
As \dpa are closed under union, this yields the desired contradiction in the form of a \dpa for $\equal_b$.

So, consider some $w \in \equal_b$, i.e., $w = w'ba^nb^n$ for some $n > 0$ and some $w' \in \set{a,b}^*$.
Define $m = \size{w'b}$ and note that $c^m w$ is in $\equal'$, i.e., accepted by $(\aut, C)$.
Hence, let $\rho$ be the unique run of $(\aut, C)$ processing $c^m w$, let $\rho'$ be the prefix of $\rho$ processing $c^m$, and let $\rho''$ be the suffix processing $w$.
So, there is some $q$ (the state $\rho'$ ends in, which is equal to the state $\rho''$ begins with) such that $\rho' \in R_q$ and the extended Parikh image~$\vec{v}\,'$ induced by $\rho'$ is in $C_q$.
Also, $\rho''$ is a run of $\aut_q$ starting in $q$ and ending in $F$.
Let $\vec{v}$ be the extended Parikh image induced by $\rho''$.
Note that $\vec{v} + \vec{v}\,'$ is the extended Parikh image induced by the full run~$\rho = \rho'\rho''$ that witnesses $c^m w \in L(\aut, C)$, and is therefore in $C$.
From this we conclude $\vec{v} \in C_q'$ and therefore that $\rho''$ is an accepting run of $(\aut_q,C_q')$ processing $w$, i.e., $w \in L(\aut_q,C_q')$ and $R_q$ is nonempty as witnessed by $\rho'$. 

For the other direction, consider a $w \in L(\aut_q,C_q')$ for some $q$ with nonempty~$R_q$. 
Then there is an accepting run~$\rho''$ of $(\aut_q,C_q')$ processing $w$, say with induced extended Parikh image~$\vec{v} \in C_q'$.
By construction, there is also a $\vec{v}\,'\in C_q$ such that $\vec{v} + \vec{v}\,' \in C$.
Furthermore, $\vec{v'}$ is the extended Parikh image induced by some run~$\rho'$ of $(\aut,C)$ processing some word of the form~$c^m$. 
Now, $\rho'\rho''$ is a run of $(\aut, C)$ starting in the initial state, processing~$c^m w$, ending in $F$, and with extended Parikh image~$\vec{v} + \vec{v}\,' \in C$, i.e., it is an accepting run.
Hence, $c^m w \in L(\aut, C) = \equal'$. 
As $w$ does not contain any $c$ ($(\aut_q,C_q')$ has no $c$-transitions), this implies that $w$ must be of the form~$w'ba^nb^n$ with $w' \in \set{a,b}^{m-1}$, i.e., $w \in \equal_b$ as required.

\ref{intitemnotuca}.) Recall that we say that $w \in \set{0,1}^*$ is non-Dyck if $\size{w}_0 < \size{w}_1$.
Now, consider the language\label{page_langdef_nondyck'}
\[
\nondyck' = \set{ c^n w \mid w=a_0 \cdots a_k \in \set{0,1}^*, \size{w} \geqslant n, \text{ and } a_0\cdots a_{n-1} \text{ is non-Dyck} },
\]
which is a variation of the language~$\nondyck$ of words that have a non-Dyck prefix.
In $\nondyck'$, the length of that prefix is given by the number of $c$'s in the beginning of the word, which makes accepting the language easier.
Nevertheless, Cadilhac et al.\ showed that $\nondyck'$ is not accepted by any \uca~\cite[Proposition~14]{CFM13}, but Bostan et al.\ showed that it is accepted by a \wupa~\cite[Remark~11]{BCKN20}.

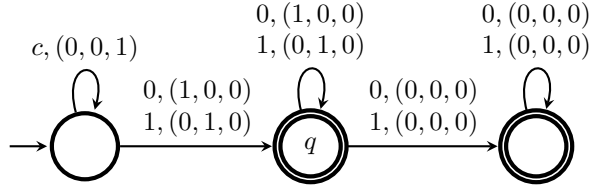
\begin{figure}[h]
    \centering
    \begin{tikzpicture}[ultra thick]
    
    \node[state] (1) at (0,0) {};
    \node[state,accepting] (2) at (3,0) {$q$};
    \node[state,accepting] (3) at (6,0) {};

    \path[-stealth]
    (-1,0) edge (1)
    (1) edge[loop above] node[above,align=left] {$c,(0,0,1)$} ()
    (1) edge node[above,align=left] {$0,(1,0,0)$\\$1,(0,1,0)$} (2)
    (2) edge[loop above] node[above,align=left]
    {$0,(1,0,0)$\\$1,(0,1,0)$} ()
    (2) edge node[above,align=left] {$0,(0,0,0)$\\$1,(0,0,0)$} (3)
    (3) edge[loop above] node[above,align=left] {$0,(0,0,0)$\\$1,(0,0,0)$} ()
;
    
    \end{tikzpicture}
    \caption{The automaton for the language $\nondyck'$.}
    \label{fig_aut-nondyckprime}
\end{figure}

In particular, it is accepted by the \pa~$(\aut, C)$ where $\aut$ is depicted in Figure~\ref{fig_aut-nondyckprime} and with
\[
C =\set{ (n,n',n+n') \mid n < n' }.
\]
Note that $(\aut, C)$ is both an \hdpa and a \wupa for $L(\aut,C)$: every word~$c^nw$ in $\nondyck'$ has at most one accepting run, the one which leaves~$q$ with the $(2n+1)$-th letter of $c^nw$ (if this letter exists).
Furthermore, this choice can be made by a resolver, as the number of $c$ at the start of the word uniquely determines when this nondeterministic choice has to be made.

\ref{intitem3}.) The (disjoint) union~$\double \cup \equal$ is accepted by some \pa, as both $\double$ and $\equal$ are accepted by a \pa and \pa are closed under union. But it is neither accepted by any \hdpa nor by any \wupa, as both models are closed under intersection\footnote{For \wupa, this was claimed by Bostan et al.~\cite{BCKN20}, for \hdpa this is shown in Theorem~\ref{thm_closure} on Page~\pageref{thm_closure}.}, i.e., $(\double \cup \equal ) \cap \set{a,b}^* = \equal$ and $(\double \cup \equal ) \cap \set{c,d}^* = \double$ yield the desired separations. 

\ref{intitem4}.) The (disjoint) union~$\nondyck' \cup \equal$ is accepted by some \wupa (as \wupa are closed under disjoint unions) but not by any \hdpa nor by any \uca, as both classes are closed under intersection. Thus, $(\nondyck' \cup \equal) \cap \set{0,1,c}^* = \nondyck'$ and $(\nondyck' \cup \equal) \cap \set{a,b}^*$ yield the desired separations.
\end{proof}

\subsection{History-deterministic Reversal-bounded Counter Machines}
\label{subsec_rbcm}

There is one more automaton model that is closely related to Parikh automata, i.e., reversal-bounded counter machines, originally introduced by Ibarra~\cite{Ibarra78RBCM}. 
These are, in their most general form, two-way automata with multiple counters that can be incremented, decremented, and tested for zero, but there is a fixed bound on the number of reversals of the reading head \emph{and} on the number of switches between increments and decrements (on each counter). 
It is known that Parikh automata and nondeterministic reversal-bounded counter machines are equivalent~\cite[Property~8]{KRfull}, while deterministic reversal-bounded counter machines are strictly more expressive than deterministic Parikh automata~\cite[Proposition~3.14]{CFM12}.
Here, we compare history-deterministic reversal-bounded counter machines and history-deterministic Parikh automata (and, for technical reasons, also history-deterministic Parikh automata with $\epsilon$-transitions).

We begin by introducing counter machines and then their reversal-bounded variant.
% Intuitively, such a machine is a read-only two-way machine with a finite number of counters holding natural numbers.
% Unlike a Parikh automaton, the transitions of a counter machine can be disabled by the counter values: More precisely, such a machine can test whether counters are zero or not.
A {(two-way) counter machine} is a tuple $\machine = (k, Q, \Sigma, \lmark, \rmark, q_\init, \Delta, F)$ where
$k \in \nats$ is the number of counters,
$Q$ is the finite set of states,
$\Sigma$ is the alphabet,
$\lmark, \rmark \notin \Sigma$ are the left and right endmarkers respectively,
$q_\init \in Q$ is the initial state,
\[\Delta \subseteq (Q \times \Sigmalr \times \set{0, 1}^k) \times (Q \times \set{-1, 0, 1} \times \set{-1, 0, 1}^k)\] 
is the transition relation, and $F \subseteq Q$ is the set of accepting states.
Here, we use the shorthand~$\Sigmalr = \Sigma \cup \set{\lmark, \rmark}$.
Intuitively, a transition~$((q, a, \vec{g}),(q',m,\vec{v}))$ is enabled if the current state is $q$, the current letter on the tape is $a$, and for each $0 \leqslant j \leqslant k-1$, the $j$-th entry in the guard~$\vec{g}$ is nonzero if and only if the current value of counter~$j$ is nonzero.
Taking this transition updates the state to $q'$, moves the head in direction~$m$, and adds the $j$-th entry of $\vec{v}$ to counter~$j$.

We require that all transitions~$((q, a, \vec{g}),(q',m,\vec{v})) \in \Delta$ satisfy the following properties:
\begin{itemize}
    \item If $a = \lmark$, then $m \geqslant 0$: the head never leaves the tape to the left.
    \item If $a = \rmark$, then $m \leqslant 0$: the head never leaves the tape to the right.
    \item $\vec{g}$ and $\vec{v}$ are \emph{compatible}, i.e. if the $j$-th entry of $\vec{g}$ is zero, then the $j$-th entry of $\vec{v}$ is nonnegative: a zero counter is not decremented.
\end{itemize}

A {configuration} of $\machine$ on an input $w\in\Sigma^*$ is of the form~$(q, \lmark w\rmark, h, \vec{c})$ where $q \in Q$ is the current state, $\lmark w \rmark$ is the content of the tape (which does not change during a run), $0 \leqslant h \leqslant \size{w} + 1$ is the current position of the reading head, and $\vec{c} \in \nats^k$ is the vector of current counter values.
The initial configuration on $w \in \Sigma^*$ is $(q_\init, \lmark w \rmark, 0, \vec{0})$,  where $\vec{0}$ is the $k$-dimensional zero vector.
    
    We say that a vector~$\vec{c} \in\nats^k$ satisfies a guard~$\vec{g} \in \set{0,1}^k$ if the following is satisfied for every $1 \leqslant j \leqslant k$: the $j$-th entry of $\vec{c}$ is zero if and only if the $j$-th entry of $\vec{g}$ is zero.
    Now, we write $(q, \lmark w\rmark, h, \vec{c}) \Rightarrow (q', \lmark w \rmark, h + m, \vec{c} + \vec{v})$ if there is a transition~$((q, a, \vec{g}), (q', m, \vec{v})) \in \Delta$ such that $\vec{c}$ satisfies $\vec{g}$, where $a$ is the $h$-th letter of $\lmark w \rmark$. 
    %By $\Rightarrow^*$ we denote the reflexive and transitive closure of $\Rightarrow$.
    
    A {run} of $\machine$ on an input~$w$ is a sequence of configurations $\run = \conf_0 \conf_1 \cdots \conf_{n}$ such that $\conf_0$ is the initial configuration on $w$ and $\conf_{j} \Rightarrow \conf_{j+1}$ for every $0 \leqslant j < n$. 
    The run~$\run$ is {accepting} if the state of $\conf_n$ is in $F$. 
    The language {accepted} by $\machine$, denoted by $L(\machine)$, is the set of words~$w \in \Sigma^*$ such that there exists an accepting run of $\machine$ on $w$.
    
    \begin{remark}
    A run on some given $w \in \Sigma^*$ is fully described by the sequence of transitions between its configurations.
    However, not every sequence of transitions induces a run.
    \end{remark}
    
    A counter machine~$\machine$ is {deterministic} if, for every $q \in Q$, $a \in \Sigmalr$ and $\vec{g} \in \set{0, 1}^k$, there is at most one triple $(q', m, \vec{v}) \in Q \times \set{-1, 0, 1} \times \set{-1, 0, 1}^k$ such that $((q, a, \vec{g}), (q', m, \vec{v})) \in \Delta$.

The problem with counter machines is that even their deterministic variants are Turing-complete \cite[Theorem~3.3]{Ibarra78RBCM}. Therefore, one must impose some restrictions in order to obtain decidability of some decision problems. One way to do that is to introduce bounds on the number of reversals of the direction the reading head moves and on the number of reversals of the stack height of each counter.

More formally, consider a run $\run$ given by a finite sequence~$((q_n, a_n, \vec {g_n}), (q_{n+1}, m_n, \vec{v_n}))_n$ of transitions. 
The {number of reversals of the reading head} during $\run$ is the number of sign alternations in the sequence~$(m_n)_n$, ignoring the $0$'s. 
The {number of reversals of the $j$-th counter} during $\run$ is the number of sign alternations in the sequence~$(v_{n, j})_n$, where $v_{n, j}$ is the $j$-th entry of $\vec{v_n}$, again ignoring the $0$'s.
A counter machine~$\machine$ is {reversal-bounded} if there exists a~$b \in \nats$ such that every accepting run has  at most $b$ reversals of the reading head and at most $b$ reversals of each counter.

% We say that a two-way counter machine is reversal-bounded, if there is a $b \in\nats$ such that on each run, the reading head reverses its direction at most $b$ times \emph{and} each counter switches between incrementing and decrementing at most $b$ times.
    We write \rbcm for reversal-bounded counter machines and 
    \nrbcm{1} for \rbcm that do not make a reversal of the reading head (i.e., they are one-way).
    Their deterministic variants are denoted by \drbcm and \ndrbcm{1}, respectively.

\begin{proposition}\label{prop:rbcm}\hfill
    \begin{enumerate}
        \item \label{prop:rbcm_1w} Every \rbcm can  be  effectively turned into an equivalent \nrbcm{1}~\cite[Theorem~2.2]{Ibarra78RBCM}.
        \item \label{prop:rbcm_counters1}
        Every \rbcm can be effectively turned into an equivalent one where the number of reversals of each counter is bounded by $1$. 
        This construction preserves determinism and one-wayness~\cite[Theorem~5]{BB74}.
        \end{enumerate}
\end{proposition}

Hence, in the following, we assume that during each run of an \rbcm, each counter reverses at most once. 

\begin{proposition}\label{prop:rbcmvspa}\hfill
\begin{enumerate}
    \item \label{prop:rbcmvspa_expr}
    \rbcm are as expressive as \pa~\cite[Property~8]{KRfull}.
    \item \label{prop:1drbcmvspa_expr}
    \ndrbcm{1} are strictly more expressive than \dpa~\cite[Proposition~3.14]{CFM12}.
\end{enumerate}
\end{proposition}

In the following, we determine the relation between history-deterministic \rbcm and \hdpa.
To this end, we first have to define the notion of history-determinism for \rbcm, which is slightly technical due to the two-wayness of these machines.

Let $\machine = (k, Q, \Sigma, \lmark, \rmark, q_\init, \Delta, F)$ be an \rbcm. 
Given a sequence~$\tau_0 \cdots \tau_j$ of transitions inducing a run~$\rho$, let $\maxpos(\tau_0 \cdots \tau_j)$ be the position of the reading head at the end of $\run$, so in particular $\maxpos(\epsilon) = 0$.
Hence, $(\lmark w \rmark)_{\maxpos(\tau_0\tau_1\cdots \tau_j)}$ is the letter the reading head is currently pointing to.
A {resolver} for $\machine$ is a function~$\hstrat\colon \Delta^* \times \Sigmalr \rightarrow \Delta$ such that if $w$ is accepted by~$\machine$, there is a sequence of transitions~$\tau_0 \tau_1 \cdots \tau_{n-1}$ such that
\begin{itemize}
    \item $\tau_{j+1} = \hstrat(\tau_0\tau_1\cdots \tau_j, (\lmark w \rmark)_{\maxpos(\tau_0\tau_1\cdots \tau_j)})$ for all $0 \leqslant j < n-1$, and 
    \item the sequence of transitions~$\tau_0 \tau_1 \cdots \tau_{n-1}$ induces an accepting run of $\machine$ on $w$.
\end{itemize}
    An \rbcm~$\machine$ is {history-deterministic} (an \hdrbcm) if there exists a resolver for $\machine$. One-way \hdrbcm are denoted by \nhdrbcm{1}.

Now, we are able to state the main theorem of this subsection: History-deterministic two-way \rbcm are as expressive as \rbcm and \pa while history-deterministic one-way \rbcm are as expressive as history-deterministic \pa.

\begin{theorem}
\label{theorem_hdrbcmvshdpa}
\hfill
\begin{enumerate}
    \item\label{theorem_hdrbcmvshdpa_twoway} \hdrbcm are as expressive as \rbcm, and therefore as expressive as \pa.
    \item\label{theorem_hdrbcmvshdpa_oneway} \nhdrbcm{1} are as expressive as \hdpa.
\end{enumerate}
\end{theorem}

\begin{proof}
The proof of the first equivalence is very general and not restricted to \rbcm: A two-way automaton over finite inputs can first read the whole input and then resolve nondeterministic choices based on the whole word. 
Spelt out more concisely: two-wayness makes history-determinism as powerful as general nondeterminism.

For the other equivalence, both directions are nontrivial: We show how to simulate a \pa using an \rbcm while preserving history-determinism, and how to simulate a \nrbcm{1} by a \pa, again while preserving history-determinism.
Due to the existence of transitions that do not move the reading head in a \nrbcm{1}, this simulation takes a detour via \pa with $\epsilon$-transitions.

\ref{theorem_hdrbcmvshdpa_twoway}.)
    It is immediate from the definition that \rbcm are at least as expressive as \hdrbcm. Thus, we show that for any \rbcm, one can construct an equivalent  \hdrbcm.
    
    Let~$\machine$ be an \rbcm. We construct an \rbcm~$\machine'$ from $\machine$ as follows: first, the reading head moves right up to the right endmarker. Then, it moves back left to the left endmarker, and then behaves like $\machine$ on the input. During the initial scan, it leaves the counters unchanged.
    It is clear that $\machine'$ is equivalent to $\machine$, as once the initial scan is finished, it behaves like $\machine$. 
    
    Furthermore, $\machine'$ is history-deterministic since, after the initial scan, which is deterministic, a resolver has access to the whole input and can, for accepted words, fix one accepting run for the input, and then resolve the nondeterminism accordingly. 
    This is possible, since the sequence of transitions seen during the initial scan is unique for every possible input, and the resolver can base its choices on that information.

\ref{theorem_hdrbcmvshdpa_oneway}.)
We first show how to turn an \hdpa into an equivalent \nrbcm{1} and then show that it is history-deterministic. 
This construction is inspired by the constructions turning a \pa into an equivalent \nrbcm{1}~\cite[Property~8]{KR} and turning a \dpa into an equivalent \ndrbcm{1}~\cite[Proposition~3.14]{CFM12}.

Given an \hdpa~$(\aut, C)$, the simulating counter machine~$\machine$ works in two phases. 
In the first phase, $\machine$ simulates the run of $(\aut, C)$ (using the same states as $(\aut, C)$) and computes the extended Parikh image of the simulated run using its counters.
Hence, in the first phase, there is a bijection between the runs of $(\aut, C)$ and $\machine$.
If the simulated run does not end in an accepting state, then the run of $\machine$ ends in a rejecting state after the first phase.
However, if the run ends in an accepting state, then the second phase begins.

In this case, acceptance of $(\aut, C)$ depends on whether the extended Parikh image of the run (stored in the counter values of $\machine$) is in $C$.
We can do so using the following result due to Cadilhac et al.~\cite[Proposition~3.13]{CFM12}: membership in a given semilinear set can be tested by a \drbcm that does not move its reading head. More formally they showed that
for every semilinear set~$C' \subseteq \nats^d$, there exists a \drbcm~$\machine_{C'}$ such that, for any configuration~$\conf_0 = (q_\init, \lmark w \rmark, i, \vec{v} \cdot \vec{0})$ with $\vec{v} \in \nats^d$, the unique run of $\machine_{C'}$ starting from $\conf_0$ is accepting if and only if $\vec{v} \in C'$. Moreover, $\machine_{C'}$ does not move its reading head.
Note that $\machine_{C'}$ will in general use more than $d$ counters.
In the configuration~$\conf_0$, the first $d$ counters contain the vector to be checked, and the remaining ones are initialized with zero. This is exactly the situation after the first phase of the simulation of the \hdpa described above has finished.
Thus, in the second phase, we run \drbcm~$\machine_C$ to check whether the extended Parikh image of the simulated run of $(\aut, C)$ stored in the counter values of the simulating run of $\machine$ is in $C$.

As there is a bijection between runs of $(\aut, C)$ and the runs of $\machine$ restricted to their first phase, it is straightforward to argue that $\machine$ accepts $L(\aut, C)$.
Furthermore, $\machine$ is one-way, as $(\aut, C)$ is one-way by definition, which implies that the first phase of the simulation is one-way as well, while the head does not move in the second phase.

It remains to argue that the resulting \nrbcm{1}~$\machine$ is indeed history-deterministic.
However, this is straightforward, since it simulates the history-deterministic \pa~$(\aut,C)$ (recall that there is a bijection between the runs, which can be used to transfer the resolver) and then deterministically checks membership in a semilinear set. 
Thus, using the bijection, a resolver for $(\aut, C)$ can directly be turned into a resolver for $\machine$.

So, let us consider the other direction, i.e., we turn an \nhdrbcm{1}~$\machine$ into an equivalent \hdpa~$(\aut, C)$.
To do so, we proceed in three steps:
\begin{description}
    \item[Step 1)] We turn $\machine$ into a \nhdrbcm{1}~$\machine'$ in a normal form that requires an \rbcm to only terminate with the reading head at the right endmarker and all counters being zero (recall that a \pa can only test for membership in the semilinear set at the end of a run).
    \item[Step 2)] We turn the \nhdrbcm{1}~$\machine'$ into an equivalent $\epsilon$-\hdpa.
    \item[Step 3)] We show how to eliminate $\epsilon$-transitions from \pa while preserving history-determinism.
\end{description}

\textbf{Step 1)} 
    A \nrbcm{1}~$\machine$ is in {normal form} if the following conditions are satisfied:
    \begin{enumerate}
        \item Let $((q, a, \vec{g}), (q', m, \vec{v}))$ be a transition of $\machine$ such that $q'$ is not final. Then, $((q, a, \vec{g}\hspace{.05cm}'), (q', m, \vec{v}))$ is also a transition of $\machine$ for every $\vec{g}\hspace{.05cm}' \in \set{0, 1}^k$ that is compatible with $\vec{v}$, i.e. $\machine$ does not test the counters during the run, but transitions that would decrement from a zero counter are still not allowed.
            
        \item Let $((q, a, \vec{g}), (q', m, \vec v))$ be a transition of $\machine$ such that $q'$ is final.
        Then, $a = \rmark$, $\vec{g} = \vec{0}$, and $\vec{v} = \vec{0}$, i.e. $\machine$ only accepts an input on the right endmarker and with all counters equal to zero.
    
        \item Accepting states do not have any outgoing transitions.
    \end{enumerate}

We show how to turn a \nrbcm{1} into an equivalent one in normal form while preserving history-determinism.

It is easy to ensure the second and third conditions, since every \nrbcm{1} can be turned into an equivalent one that, upon reaching an accepting state, deterministically moves to the right endmarker and empties its counters, and then reaches a fresh accepting state without outgoing transitions.
This also preserves history-determinism. 

For the first condition, let $\machine$ be a \nrbcm{1} meeting the other two conditions. 
Also, recall that we assume, again without loss of generality, that every counter reverses at most once. 
We construct $\machine'$ that simulates $\machine$ and instead of testing counters for zero it guesses the outcome of a test and verifies the guesses at the end of the run. 
The ability to postpone these tests crucially relies on the fact that every counter reverses at most once. 

For every counter, $\machine'$ stores in its states one of the following statuses:
\begin{description}
    \item[\stini] the counter has neither been increment nor decremented.
    \item[\stinc] the counter has been incremented but not yet decremented.
    \item[\stdec] the counter has been incremented, decremented, but is still assumed to be nonzero.
    \item[\stzero] the counter has been guessed to be zero.
\end{description}

At the beginning, all the counters are in status \stini. Then, when a counter has status \stini or \stinc and is incremented, then, its new status is \stinc.
If the status is \stinc or \stdec and the counter is decremented, its status is nondeterministically updated to \stdec or \stzero. 
Intuitively, it should be updated to \stdec as long as the true counter value is nonnegative and to \stzero if the true value is zero. 
A counter with status \stzero cannot be decremented (as it is assumed to be zero) nor incremented (as there is at most one reversal, and the status cannot be changed anymore.
This implies that the value of the corresponding counter can no longer be updated as soon as its status is changed to \stzero. 

Now, $\machine'$ simulates $\machine$, accounting for tests in guards as follows:
    \begin{itemize}
        \item If a counter has status \stini or \stzero, then its value is assumed to be equal to zero.
        \item If a counter has status \stinc or \stdec, then its value is assumed to be positive.
    \end{itemize}
As a counter is no longer updated once its status is \stzero, the guess that the counter is zero is correct if and only if the value of the counter is indeed zero at the end of the run.

As $\machine$ satisfies the second condition of the normal form, every accepting run ends with a transition testing all counters for zero. Thus, an accepting state of $\machine$ is only reached if all counters are zero (and stay zero after the transition has been taken).
So, we equip $\machine'$ with a transition that checks, after the simulation of $\machine$ has finished in an accepting state, whether all counters are zero and have status~\stini or \stzero.
It only accepts if this is the case.

Therefore, $\machine'$ accepts the same language as $\machine$. In addition, the only nondeterminism introduced in the construction are the guesses whether a counter is zero after a decrement. These  can be made history-determinstically since a resolver can keep track of the current values of the counters. 
Hence, $\machine'$ is history-deterministic if $\machine$ is.

\textbf{Step 2)}
Here, we show how to turn a \nrbcm{1} into an equivalent $\epsilon$-\hdpa while preserving history-determinism.
So, let us first introduce the latter type of automaton.

An $\epsilon$-\pa over an alphabet~$\Sigma$ and of dimension~$d$ is a tuple~$(\aut, C)$, where $\aut$ is a finite automaton over $(\Sigma \cup \set{\epsilon}) \times D$, where $\epsilon$ is the empty word, and where $C \subseteq \nats^d$ is semilinear.
Note that this definition does not coincide with the classical notion of $\epsilon$-\nfa, as $\epsilon$-transitions in $\aut$ are still labeled by a vector in the second component.
The language accepted by $(\aut, C)$ is $L(\aut, C) = \set{ \projsigma(w) \mid w \in L(\aut) \text{ and } \parikhimage(w) \in C }$.
Here, we treat $\projsigma$ as a homomorphism, which implies $\projsigma(\epsilon) = \epsilon$.
Note that $\epsilon$-\pa are not more expressive than \pa, since \pa are closed under homomorphisms: Hence, one can treat $\epsilon$ as a ordinary letter and then apply a homomorphism that deletes exactly that letter. 

Let $(\aut, C)$ be an $\epsilon$-\pa with $\aut = (Q, (\Sigma \cup \set{\epsilon}) \times D, q_\init, \Delta, F)$. 
Let $\Delta_\epsilon$ be the set of $\epsilon$-transitions and $\Delta_\Sigma$ be the set of $\Sigma$-transitions of $\aut$.
A {resolver} for $(\aut, C)$ is a function $\hstrat: \Sigma^+ \rightarrow \Delta_\epsilon^* \cdot \Delta_\Sigma \cdot \Delta_\epsilon^*$ such that the image~$\hstratiter(w)$ of $r$'s iteration (defined analogously to the case of Parikh automata without $\epsilon$-transitions in Section~\ref{sec_hd}) is an accepting run processing $w$, for every $w \in L(\aut, C)$.
An $\epsilon$-\pa $(\aut, C)$ is said to be {history-deterministic} (an $\epsilon$-\hdpa) if there exists a resolver for it.

Now, fix a \nrbcm{1}~$\machine$ in normal form, say with $k$ counters. We construct an equivalent $\epsilon$-\pa~$(\aut, C)$ and show that the transformation preserves history-determinism.

Intuitively, $\aut$ simulates $\machine$ and the semilinear set is used to check that all counters of $\machine$ are zero after the simulation has ended.
However, recall that a \pa can only increment its counter while an \rbcm can increment and decrement its counters. 
Hence, $\aut$ has two counters for each counter~$j$ of $\machine$, one (counter $2j$) counting the increments and the other one (counter $2j+1$) counting the decrements during the simulation.
 As $\machine$ does not test its counters during a run (due to the second condition of the normal form), $\aut$ only has to test whether the counters are equal to zero in the last configuration of a simulated run.
 This is the case if and only if the value of counter $2j$ of $(\aut,C)$ is equal to the value of counter $2j+1$, for every counter~$j$ of machine.
 This can easily be expressed by a semilinear set.
   
As $(\aut, C)$ simulates a run of $\machine$ and accepts if all the counters of $\machine$ are equal to zero at the end, they both recognise the same language. In addition, this transformation does not add any nondeterminism as the \pa simply simulates the run of the \nrbcm{1}. Therefore, $(\aut, C)$ is history-deterministic if $\machine$ is.

\textbf{Step 3)}
Finally, we need to show how to eliminate $\epsilon$-transitions in $\epsilon$-\pa while preserving history-determinism.
The construction is similar to one used by Klaedtke and Rueß~\cite[Theorem~20]{KRfull} to show that Parikh automata are closed under homomorphisms, which also requires the removal of $\epsilon$-transitions. 
For the sake of completeness, we present the construction here, as we need to show that it allows us to preserve history-determinism.

Let $(\aut, C)$ be an $\epsilon$-\pa of dimension~$d$ and let $\varphi(x_0, \dots, x_{d-1})$ be a Presburger formula defining $C$ (recall Proposition~\ref{prop_presburger}). 
    We construct an equivalent \pa~$(\aut', C')$ by replacing paths that are labeled by $\epsilon^* a \epsilon^*$ by a single $a$-transition.
    However, taking $\epsilon$-transitions has side-effects, i.e., the extended Parikh image is updated along them.
    This requires us to account for the effect of taking $\epsilon$-cycles in the new semilinear set~$C'$.
     To do so correctly, we need to keep track of (in the extended Parikh image) which $\epsilon$-cycles could have been traversed during the run of $(\aut, C)$ being simulated by $(\aut', C')$. 
     
A finite run infix~$\rho$ is called reduced if no state is repeated in it. 
Let $a \in \Sigma$. 
An $\epsilon$-$a$-$\epsilon$-path is a run infix of the form~$\rho_0 \tau \rho_1$ where $\rho_0$ and $\rho_1$ are (possibly empty) sequences of $\epsilon$-transitions and where $\tau$ is a transition labeled by $a$.
We say that $\rho_0 \tau \rho_1$ is reduced if $\rho_0$ and $\rho_1$ are reduced.
Note that there are only finitely many reduced $\epsilon$-$a$-$\epsilon$-paths for every $a$.
Finally, let $\set{K_d, K_{d+1}, \ldots, K_{m}}$ be the set of $\epsilon$-cycles of $(\aut, C)$, which is again finite as we only consider simple cycles. 

Consider an arbitrary, not necessarily reduced, $\epsilon$-$a$-$\epsilon$-path $\rho_0 \tau \rho_1$. 
If it is not reduced, then $\rho_0$ or $\rho_1$ contains an $\epsilon$-cycle~$K_j$.
Removing this cycle yields a shorter $\epsilon$-$a$-$\epsilon$-path starting and ending in the same state (but possibly with different extended Parikh image).
By repeating this operation (always removing the first cycle in case there are several ones), we turn every $\epsilon$-$a$-$\epsilon$-path~$\rho$ into a reduced $\epsilon$-$a$-$\epsilon$-path~$\red{\rho}$ that still has the same first state and the same last state as the original one, but possibly a different extended Parikh image.

Now, we define $\aut'$ as follows: it has the same states, the same initial state, the same set of accepting states as $\aut$, and the same alphabet~$\Sigma$. 
Furthermore, for every $\epsilon$-$a$-$\epsilon$-path
\[
\rho = (q_0, a_0, q_1)\cdots (q_{n-1}, a_{n-1}, q_n)
%, \text{where $a_i \in \Sigma \cup \{\epsilon\}$ for $0 \leqslant i < n$},
\]
in $\aut$, the automaton~$\aut'$ has the transition~$(q_0, (a, \vec{v})\cdot \cycleregister(\rho)),q_n)$, where $\vec{v}$ is the extended Parikh image of the reduced run~$\red{\rho}$ of $\rho$, and where 
\[
\cycleregister(\rho) = (b_d,b_{d+1}, \ldots, b_m)
\]
is such that $b_j$ is equal to $1$ if $K_j$ has been removed at least once from $\rho$ to obtain $\red{\rho}$.
Otherwise, $b_j$ is equal to $0$. 
Note that this results in finitely many transitions, as there are only finitely many reduced $\epsilon$-$a$-$\epsilon$-paths and only finitely many choices for the $\cycleregister(\rho)$.\footnote{Furthermore, it is not hard to see that the construction can be made effective, as one does not have to consider all $\epsilon$-$a$-$\epsilon$-paths: every transition in $\aut'$ is witnessed by an $\epsilon$-$a$-$\epsilon$-path of bounded length.}

We then define $C'$ by the Presburger formula \[\varphi'(x_0, \dots, x_{m}) = \exists y_d\exists y_{d+1} \cdots \exists y_m \bigwedge_{j = d}^m (x_j >0 \Leftrightarrow y_j>0) \wedge \varphi(
(x_0, \ldots, x_{d-1}) +\sum_{j=d}^m y_j\cdot \vec{v}_j
)
\] where $\vec{v}_j$ is the extended Parikh image of $K_j$. 
The idea is that if a run could have went through the cycle~$K_j$, then $y_j$ captures how many times it would have been used.
Note that this formula is not a first-order formula, due to the summation of the vectors, but can easily be turned into one by making the addition componentwise.
We refrain from doing so for the sake of readability.
Also, the multiplication in the formula is not an issue, as we only multiply with constants, which is just repeated addition.

Next, we show that $(\aut, C)$ and $(\aut',C')$ are equivalent.
Let us first remark that either both accept the empty word or both reject it. This is because they share the same initial state and same accepting states and because the zero vector is in $C$ if and only if it is in $C'$.
So, in the following we only have to consider nonempty words.

Let $w=a_0 \cdots a_{n-1} \in L(\aut, C)$ with $n>0$, say with accepting run $\rho$.
Then, $\rho$ can be decomposed into a sequence~$\rho_0 \cdots \rho_{n-1}$, where each $\rho_i$ is an $\epsilon$-$a_i$-$\epsilon$-path (the exact splits are irrelevant).
Consider one of these $\rho_i$, say it leads from $q$ to $q'$.
By construction, $(\aut', C')$ has an $a_i$-transition from $q$ to $q'$.
These transitions form a run~$\rho'$ of $(\aut',C')$ processing $w$. Note that both runs end in the same state, so $\rho'$ ends in an accepting state as well.

Now, $\cycleregister(\rho_i)$ encodes which cycles have been removed from $\rho_i$ to obtain $\red{\rho_i}$.
Hence, taking the first $d$ components of the extended Parikh image of $\rho'$ and adding to it the extended Parikh images of the removed cycles (with their appropriate multiplicity~$y_j > 0$) yields the Parikh image of $\rho$, which is in $C$.
Hence, the extended Parikh image of $\rho'$ is in $C'$.

Now, assume $a_0 \cdots a_{n-1} \in L(\aut', C')$ with $n>0$, say with accepting run $\rho' = \tau_0 \cdots \tau_{n-1}$.
The transition~$\tau_i$ is witnessed by some reduced $\epsilon$-$a_i$-$\epsilon$-path $\rho_i$.
These form a run~$\rho$ of $\aut$ processing $w$, which ends in an accepting state, as $\rho'$ does.
Furthermore, we know that the extended Parikh image of $\rho'$ is in $C'$, so there are values $y_j > 0$ for the $\epsilon$-cycles~$K_j$ such that $(x_0, \ldots, x_{d-1} +\sum_{j=d}^m y_j\cdot \vec{v}_j)$ is in $C$. 
Hence, we can add each $K_j$ $y_j$ times to $\rho$ in order to obtain another accepting run of $(\aut,C)$ that processes $w$ and has an extended Parikh image in $C$.

Finally, we need to show that $(\aut', C')$ is history-deterministic if $(\aut, C)$ is.
So, let $\hstrat\colon \Sigma^+ \rightarrow \Delta_\epsilon^* \Delta_\Sigma\Delta_\epsilon^*$ be a resolver for $(\aut, C)$.
We construct a resolver~$\hstrat' \colon \Sigma^+ \rightarrow \Delta'$ for $(\aut',C')$.
We can assume without loss of generality that $\hstrat(wa)$ is an $\epsilon$-$a$-$\epsilon$-path for every input~$wa$: if not, then this output cannot be part of an accepting run, which means we can redefine $\hstrat$ arbitrarily so that it satisfies our assumption.

Now, let $wa \in \Sigma^+$ and let $\rho = \hstrat(wa)$, which is an $\epsilon$-$a$-$\epsilon$-path, say from $q$ to $q'$.
The reduced $\epsilon$-$a$-$\epsilon$-path~$\red{\rho}$ leads from $q$ to $q'$ and induces a transition~$\tau$ of $\aut'$ processing $a$.
We define $\hstrat'(wa) = \tau$ for this transition.
Then, an induction on the length of an input word~$w$ shows that the run of $(\aut', C')$ constructed by $\hstrat'$ on an input $w$ is accepting if and only if the run of $(\aut, C)$ constructed by $\hstrat$ on $w$ is accepting.
\end{proof}

Finally, let us remark that \hdpa (or equivalently \nhdrbcm{1}) and deterministic \rbcm have incomparable expressiveness.
Indeed, the language $\equal$, which is not accepted by any \hdpa (see Theorem~\ref{thm_sep}), can easily be accepted by a deterministic \rbcm while the language~$\nondyck$ (see Example~\ref{example_nondyck}) is accepted by an \hdpa, but not by any deterministic \rbcm.
The reason is that these machines are closed under complement, but the complement of $\nondyck$ is not accepted by any \pa~\cite[Proposition~11]{CFM13}, and therefore also not by any \rbcm.

\section{Closure Properties}
\label{subsec_closure}

In this subsection, we study the closure properties of history-deterministic Parikh automata, i.e., we consider Boolean operations, concatenation and Kleene star, (inverse) homomorphic image, and commutative closure.
Let us begin by recalling the last three notions.

Fix some alphabet~$\Sigma = \set{a_0, a_1, \cdots, a_{d-1}}$ with fixed order~$a_0 < a_1 < \cdots < a_{d-1}$. 
The Parikh image of a word~$w \in \Sigma^*$ is the vector~$\Phi(w) = (\occ{w}{a_0},\occ{w}{a_1}, \ldots, \occ{w}{a_{d-1}})$ and the Parikh image of a language~$L \subseteq \Sigma^*$ is $\Phi(L) = \set{\Phi(w) \mid w \in L}$.
The commutative closure of $L$ is $\set{w \in \Sigma^* \mid \Phi(w) \in \Phi(L)}$.

Now, fix some alphabets~$\Sigma$ and $\Gamma$ and a homomorphism~$h \colon \Sigma^* \rightarrow \Gamma^*$. 
The homomorphic image of a language~$L \subseteq \Sigma^*$ is $h(L) = \set{h(w) \mid w \in L} \subseteq \Gamma^*$. 
Similarly, the inverse homomorphic image of a language~$L \subseteq \Gamma^*$ is $h^{-1}(L) = \set{w \in \Sigma^* \mid h(w) \in L}$.

\begin{theorem}
\label{thm_closure}
\hdpa are closed under union, intersection, inverse homomorphic images, and commutative closure, but not under complement, concatenation, Kleene star, and homomorphic image. 
\end{theorem}

\begin{proof}
For $i \in \set{1,2}$, let $(\aut_i, C_i)$ be an \hdpa with $\aut_i = (Q_i, \Sigma \times D_i, q_\init^i, \Delta_i, F_i)$, say with resolver~$\hstrat_i$. 
Furthermore, let $d_i$ be the dimension of $D_i$.
As in the proof of Lemma~\ref{lemma_pumping}, we assume without loss generality that $\hstratiter_i(w)$ is a run of $\aut_i$ processing $w$, for every $w \in \Sigma^*$.

First, we consider closure under union. 
Intuitively, we use a product construction to simulate a run of $\aut_1$ and a run of $\aut_2$ simultaneously. 
A naive approach would be to take the classical product of the $\aut_i$ where we concatenate the vectors labelling the transitions, and then use $C = C_1 \cdot \nats^{d_2} \cup \nats^{d_1} \cdot C_2$.
However, this is not correct, as this automaton can accept if an accepting state of $\aut_1$ is reached while the extended Parikh image is in $\nats^{d_1} \cdot C_2$ or vice versa.
To overcome this issue, we reflect in the extended Parikh image  which one of the simulated runs ends in an accepting state. 
To simplify this, we assume, without loss of generality, that the initial states of both $\aut_i$ do not have incoming transitions. 

Now, we define the product~$\aut = (Q_1 \times Q_2, \Sigma \times D, (q_\init^1, q_\init^2), \Delta, Q_1 \times Q_2)$ where
\begin{itemize}
    \item $D = \set{1,2} \cdot \set{1,2} \cdot D_1 \cdot D_2$, and
    \item $\Delta = \set{ ( (q_1, q_2),(a, (f_1, f_2) \cdot\vec{v}_1 \cdot \vec{v}_2),(q_1', q_2') ) \mid ( q_i, (a, \vec{v}_i),q_i' ) \in \Delta_i \text{ for } i \in\set{1,2} }$, where $f_i$ for $i \in\set{1,2}$ is defined as follows.
    \begin{itemize}
        \item If $q_i$ is the initial state of $\aut_i$: If $q_i' \in F_i$ then $f_i = 2$, otherwise $f_i = 1$.
        \item If $q_i$ is not the initial state of $\aut_i$: If $q_i \in F_i \Leftrightarrow q_i' \in F_i$ then $f_i = 2$, otherwise $f_i = 1$.
    \end{itemize}
    Note that this satisfies the following invariant for every nonempty run of $\aut$: If $w \in (\Sigma\times D)^+$ is the input processed by the run, $\parikhimage(w)$ = $(f_1, f_2, \ldots)$, and the run ends in state~$(q_1,q_2) \in Q_1 \times Q_2$, then we have $f_i \bmod 2 = 0$ if and only if $q_i \in F_i$.
    Thus, it is the values~$f_i$ that reflect in the extended Parikh image which of the simulated runs end in an accepting state. 
    Note however, that this only holds for nonempty runs, as we need to initialize the reflection.
\end{itemize}

Let $C_\epsilon$ be the set containing the  $(2+d_1+d_2)$-dimensional zero vector if $\epsilon \in L(\aut_1, C_1) \cup L(\aut_2, C_2)$, and $C_\epsilon$ be the empty set otherwise.
Then, we define
\[
C = \set{n \in\nats \mid n \bmod 2 = 0 } \cdot \nats \cdot  C_1 \cdot \nats^{d_2}\,\cup
 \nats \cdot \set{n \in\nats \mid n \bmod 2 = 0 } \cdot  \nats^{d_1} \cdot  C_2 \cup C_\epsilon,
\]
which is semilinear due to Proposition~\ref{propsemilinearclosure}.
Then, we have $L(\aut, C) = L(\aut_1, C_1) \cup L(\aut_2, C_2)$ and the following function~$\hstrat$ is a resolver for $(\aut, C)$: let $\hstrat_i(w) = ( q_i, (a, \vec{v}_i),q_i' )$ and define
$\hstrat(w) = ( (q_1, q_2),(a, (f_1,f_2)\cdot\vec{v}_1 \cdot \vec{v}_2),(q_1', q_2') )$, where $(f_1,f_2)$ is defined as above.

Now, consider closure under intersection.
Here, we take the product~$\aut' = (Q_1 \times Q_2, \Sigma \times (D_1 \cdot D_2), (q_\init^1, q_\init^2), \Delta', F_1 \times F_2)$ with
\[\Delta' = \set{ ( (q_1, q_2),(a, \vec{v}_1 \cdot \vec{v}_2),(q_1', q_2') ) \mid ( q_i, (a, \vec{v}_i),q_i' ) \in \Delta_i \text{ for } i \in\set{1,2} }\]
and define $C' = C_1 \cdot C_2$, which is is semilinear due to Proposition~\ref{propsemilinearclosure}. 
Then, $L(\aut', C') = L(\aut_1, C_1) \cap L(\aut_2, C_2)$ and the following function~$\hstrat$ is a resolver for $(\aut, C)$: let $\hstrat_i(w) = ( q_i, (a, \vec{v}_i),q_i' )$ and define
$\hstrat(w) = ( (q_1, q_2),(a, \vec{v}_1 \cdot \vec{v}_2),(q_1', q_2') )$.

Now, consider closure under inverse homomorphic images.
Klaedtke and Rueß have shown that \dpa and \pa are effectively closed under inverse homomorphic images~\cite[Property~4]{KR}.
We follow their construction to show that \hdpa are closed under inverse homomorphic images.
In fact, a similar construction is also used to show that regular languages are also closed under inverse homomorphic images~\cite{Hopcroft}.

Consider a homomorphism $h: \Sigma^* \rightarrow \Gamma^*$.
Given an \hdpa $(\aut, C)$ with $\aut=(Q,\Gamma \times D, q_\init, \Delta, F)$, we construct another \hdpa $(\aut', C)$ with $\aut'=(Q,\Sigma \times D', q_\init, \Delta', F)$, such that $L(\aut', C) = h^{-1}(L(\aut, C))$.
Note that the set $C$ is the same in both automata.
Intuitively, $\aut'$ processes a letter~$(a, \vec{v})$ by simulating a sequence of transitions in $\aut$ processing $(b_1, \vec{v}_1), \dots ,(b_m, \vec{v}_m)$ with $h(a)=b_1 \cdots b_m$, and $\vec{v} = \sum_{i=1}^{m}\vec{v}_i$, for $m \in \nats$.
The \nfa~$\aut'$ has the same set $Q$ of states as $\aut$, and $\aut'$ has a transition from state~$p$ to state~$q$ with label~$(a,\vec{v}) \in \Sigma \times D'$ if and only if there is a sequence of transitions labeled with $(b_1, \vec{v}_1), \dots ,(b_m, \vec{v}_m) \in (\Gamma \times D)^*$ taking $\aut$ from state $p$ to state $q$, and $h(a)=b_1 \cdots b_m$, and $\vec{v} = \sum_{i=1}^{m}\vec{v}_i$.
It is easy to see that the set $D'$ thus obtained is finite, and that $(\aut', C)$ accepts the language $h^{-1}(L(\aut, C))$.

We are now left to prove that $(\aut', C)$ is an \hdpa.
For notational convenience, we lift the definition of $\parikhimage$ to the runs of a PA as $\parikhimage(\run) = \parikhimage(w)$, where $w$ is the word processed by $\run$.
Since $(\aut, C)$ is an \hdpa, there exists a resolver $\hstrat$ of $(\aut, C)$.
We define a resolver $\hstrat'$ for $(\aut', C)$ as follows.
Consider a word $w \in L(\aut', C)$, and let $x=h(w)$.
We have that $x \in L(\aut, C)$.
% Let $w = w_1 \dots w_t$, where each $w_i$ is a letter in $\Gamma$ for $1 \leqslant \ell \leqslant t$, and let $r(w_1 \dots w_\ell) = \delta_\ell$.
Let $w = a_1 \cdots a_m$ such that $a_j \in \Sigma$ for $1 \leqslant j \leqslant m$.
We have that $h(w) = h(a_1) \cdots h(a_m)$, and let $h(a_j) = x_{j,1} \cdots x_{j,k_j}$ such that each $x_{j,i} \in \Gamma$ for $1 \leqslant i \leqslant k_j$.
We define $\delta_{j,i} = r(h(a_1) \cdots h(a_{j-1}) x_{j,1} \cdots x_{j,i})$, for $1 \leqslant j \leqslant m$ and $1 \leqslant i \leqslant k_j$, to be the transition induced by $\hstrat$ after processing the prefix $h(a_1) \cdots h(a_{j-1}) x_{j,1} \cdots x_{j,i}$ of $h(w)$.
Further, let $\delta_{j,i} = (q_{j, i-1}, (x_{j,i}, \vec{v}_{j,i}), q_{j,i})$.
In particular, we have that $q_{j,0}$ is the state in $\aut$ before processing $x_{j,1}$ and $q_{j,k_j}$ is the state in $\aut$ after processing $x_{j,k_j}$.
Then, we define $\hstrat'(a_1 \cdots a_j) = (q_{j,0}, (a_j, \vec{v}_j), q_{j,k_j})$, where $\vec{v}_j=\sum_{i=1}^{k_j} \vec{v}_{j,i}$.

% Now we define $\hstrat'(w_1 \dots w_j) = \tau(\delta_{j,1} \cdots \delta_{j,k_j})$, where $\tau: \Delta^* \rightarrow \Delta'$ is a partial function defined as 
% \[
% \tau(\delta_1 \dots \delta_\ell) = \delta'(q_0,(h^{-1}(a_0 \dots a_{\ell-1}), \vec{v}), q_p\ell)
% \]
% with $\delta_i = (q_{i-1},(a_{i-1}, \vec{v}_{i-1}), q_i)$, $\vec{v}=\sum_{i=0}^{\ell-1} \vec{v}_i$, and $\ell \geqslant 1$.
% The function $\tau$ is partial since $h^{-1}$ is a partial function.
Recall that $w \in L(\aut', C)$ and $h(w)=x \in L(\aut, C)$.
If $q_f \in F$ is the state reached after processing $x$ in $(\aut, C)$ following the accepting run $\hstratiter(x)$, then the same state $q_f$ is reached after processing $w$ in $(\aut', C)$ in the run ${r'}^*(w)$.
Besides, the extended Parikh images corresponding to both $\hstratiter(x)$ and ${r'}^*(w)$ are the same, that is $\parikhimage(\hstratiter(x)) = \parikhimage({r'}^*(w))$.
Now since $\parikhimage(\hstratiter(x))$ belongs to $C$, we have that $\parikhimage({r'}^*(w))$ also belongs to $C$, and thus ${r'}^*(w)$ is an accepting run.
Hence $\hstrat'$ is indeed a resolver for $(\aut', C)$.

Now, let us consider commutative closure. Cadilhac et al.\ proved that the commutative closure of any \pa (and therefore that of any \hdpa) is accepted by some \dpa~\cite[Proposition~3.17]{CFM12}, and therefore also by some \hdpa.

The negative results follow from a combination of expressiveness results proven in Section~\ref{sec_expressiveness} and nonexpressiveness results in the literature:
\begin{itemize}
\item Complement: In the first part of the proof of Theorem \ref{thm_sep}, we show that the language~$\nondyck$ is accepted by an \hdpa, but its complement is known to not be accepted by any \pa~\cite[Proposition~11]{CFM13}.

\item Concatenation: The language~$\equal$ is the concatenation of the languages~$\set{a,b}^*$ and $\set{a^n b^n \mid n \in \nats}$, which are both accepted by a \dpa (and therefore also by \hdpa), but itself is not accepted by any \hdpa (see the proof of Theorem~\ref{thm_sep}). 

\item Kleene star: There is a \dpa (and therefore also an \hdpa) such that the Kleene star of its language is not accepted by any \pa~\cite[Proposition~3.17]{CFM12}, and therefore also by no \hdpa.

\item Homomorphic image: The language
\[
\set{a,b}^*\cdot \set{c a^{n-1}b^n \mid n>0}
\]
is accepted by the \dpa~$(\aut, C)$ with $\aut$ as in Figure~\ref{fig_autmorphism} and $C = \set{(n,n) \mid n > 0}$, and thus also by an \hdpa.
\begin{figure}[h]
    \centering
    \begin{tikzpicture}[ultra thick]
    
    \node[state] (1) at (0,0) {};
    \node[state] (2) at (3,0) {};
    \node[state,accepting] (3) at (6,0) {};
    
    \path[-stealth]
    (-1,0) edge (1)
    (1) edge[loop above] node[above,align=left] {$a,(0,0)$\\ $b,(0,0)$} ()
    (1) edge node[above] {$c,(1,0)$} (2)
    (2) edge[loop above] node[above] {$a,(1,0)$} ()
    (2) edge node[above] {$b,(0,1)$} (3)
    (3) edge[loop above] node[above] {$b,(0,1)$} (3);
    
    \end{tikzpicture}
    \caption{The automaton for the language $\set{a,b}^*\cdot \set{c a^{n-1}b^n \mid n>0}$.}
    \label{fig_autmorphism}
\end{figure}

However, its homomorphic image under the morphism~$h$ uniquely identified by $h(a) = h(c) = a$ and $h(b)= b$ is $\equal$, which is not accepted by any \hdpa, as shown in Theorem~\ref{thm_sep}.\qedhere
\end{itemize}
\end{proof}

Table~\ref{tab:closure} compares the closure properties of \hdpa with those of \dpa, \uca, and \pa. 
We do not compare to \rbcm, as only deterministic ones differ from Parikh automata and results on these are incomplete:
However, Ibarra proved closure under union, intersection, and complement~\cite[Lemma~3.1 and Lemma~3.2]{Ibarra78RBCM}.

\begingroup

\begin{table}[h]
    \caption{Closure properties of history-deterministic Parikh automata (in grey) and comparison to other types of Parikh automata (results for other types are from~\cite{CFM12,CFM13,KR}).}
    \label{tab:closure}
    \centering
\begin{tabular}{lcccccccc}
\toprule
 & union & intersection &  complement &  concate- & Kleene  &  homo- & inverse  & commutative \\
  & &  &  &nation&  star& morphism &  homomorphism &  closure\\
 % & $\cup$ & $\cap$ &  $\overline{\phantom{x}}$ &  $ \cdot$ & $^*$ &  $h $ & $h^{-1}$ & c\\
 \midrule
\dpa  & Y & Y & Y & N & N & N & Y & Y  \\
\rowcolor{lightgray!50}\hdpa & Y & Y & N & N & N & N & Y & Y \\
\uca & Y & Y & Y & N & N & N & ? & Y \\
\pa  & Y & Y & N & Y & N & Y & Y & Y \\

\bottomrule
\end{tabular}
\end{table}
\endgroup

\section{Decision Problems}
\label{subsec_decision}
Next, we study various decision problems for history-deterministic \pa. 
First, let us mention that nonemptiness and finiteness are decidable for \hdpa, as these problems are decidable for \pa~\cite[Property~6]{KR}\cite[Proposition~3.16]{CFM12}. 
In the following, we consider universality, inclusion, equivalence, regularity, and model checking.

Our undecidability proofs are reductions from nontermination problems for two-counter machines, using an encoding of two-counter machines by \dpa originally developed for \pa over infinite words~\cite{GJLZ22}. 

A two-counter machine~$\mach$ is a sequence
\[
(0:  \instr_0) (1:  \instr_1) \cdots (k-2:  \instr_{k-2})(k-1:  \stopp), 
\]
where the first element of a pair~$(\ell: \instr_\ell)$ is the line number and $\instr_\ell$ for $0 \leqslant\ell < k-1$ is an instruction of the form
\begin{itemize}
    \item $\inc{i}$ with $i \in\set{0,1}$, 
    \item $\dec{i}$ with $i \in\set{0,1}$, or
    \item $\ite{i}{\ell'}{\ell''}$ with $i \in\set{0,1}$ and $\ell',\ell'' \in \set{0, \ldots,k-1}$. 
\end{itemize}
A configuration of $\mach$ is of the form~$(\ell, c_0, c_1)$ with $\ell \in \set{0, \ldots, k-1}$ (the current line number) and $c_0, c_1\in\nats$ (the current contents of the counters). 
The initial configuration is~$(0,0,0)$ and the unique successor configuration of a configuration~$(\ell,  c_0, c_1)$ is defined as follows:
\begin{itemize}
    \item If $\instr_\ell = \inc{i}$, then the successor configuration is $(\ell +1, c_0', c_1')$ with $c_i' = c_i +1$ and $c_{1-i}' = c_{1-i}$.
    \item If $\instr_\ell = \dec{i}$, then the successor configuration is $(\ell +1, c_0', c_1')$ with $c_i' = \max\set{c_i -1,0}$ and $c_{1-i}' = c_{1-i}$.
    \item If $\instr_\ell = \ite{i}{\ell'}{\ell''}$ and $c_i = 0$, then the successor configuration is $(\ell', c_0, c_1)$.
    \item If $\instr_\ell = \ite{i}{\ell'}{\ell''}$ and $c_i > 0$, then the successor configuration is $(\ell'', c_0, c_1)$.
    \item If $\instr_\ell = \stopp$, then $(\ell, c_0, c_1)$ has no successor configuration.
\end{itemize}
The unique run of $\mach$ (starting in the initial configuration) is defined as expected.
It is either finite (line~$k-1$ is reached) or infinite (line~$k-1$ is never reached).
In the former case, we say that $\mach$ terminates.

\begin{proposition}[\cite{Minsky67}]
The following problem is undecidable: Given a two-counter machine~$\mach$, does $\mach$ terminate?
\end{proposition}

In the following, we assume without loss of generality that each two-counter machine satisfies the \emph{guarded-decrement property}:
Every decrement instruction~$(\ell: \dec{i})$ is preceded by $(\ell-1: \ite{i}{\ell+1}{\ell})$ and decrements are never the target of a goto instruction. 
As the decrement of a zero counter has no effect, one can modify each two-counter machine~$\mach$ into an $\mach'$ satisfying the guarded-decrement property such that $\mach$ terminates if and only if $\mach'$ terminates: One just adds the the required guard before every decrement instruction and changes each target of a goto instruction that is a decrement instruction to the preceding guard.

The guarded-decrement property implies that decrements are only executed if the corresponding counter is nonzero. 
Thus, the value of counter~$i$ after a finite sequence of executed instructions (starting with value zero in the counters) is equal to the number of executed increments of counter~$i$ minus the number of executed decrements of counter~$i$. 
Note that the number of executed increments and decrements can be tracked by a Parikh automaton.

Consider a finite or infinite word~$w = a_0 a_1 a_2 \cdots$ over the set~$\set{0,1,\ldots, k-1}$ of line numbers.
We now describe how to characterize whether $w$ is (a prefix of) the projection to the line numbers of the unique run of $\mach$ starting in the initial configuration.
This characterization is designed to be checkable by a Parikh automaton.
Note that $w$ only contains line numbers, but does not encode values of the counters. These will be kept track of by the Parikh automaton by counting the number of increment and decrement instructions in the input, as explained above (this explains the need for the guard-decrement property).
Formally, we say that $w$ contains an \emph{error} at position~$n < \size{w}-1$ if either $a_n = k-1$ (the instruction in line~$a_n$ is $\stopp$), or if one of the following two conditions is satisfied:
    \begin{enumerate}
        \item\label{casenonif} The instruction~$\instr_{a_n}$ in line~$a_{n}$ of $\mach$ is an increment or a decrement and $a_{n+1} \neq a_{n}+1$, i.e., the letter $a_{n+1}$ after $a_n$ is not equal to the line number $a_n + 1$, which it should be after an increment or decrement.
        
        \item\label{caseif} $\instr_{a_{n}}$ has the form~$\ite{i}{\ell}{\ell'}$, and one of the following cases holds: Either, we have
            \[
            \sum\nolimits_{j  \colon \instr_j = \inc{i}} \occ{a_0 \cdots a_{n}}{j} =
            \sum\nolimits_{j \colon \instr_j = \dec{i}} \occ{a_0 \cdots a_{n}}{j}
            \] and $a_{n+1} \neq \ell$,
            i.e., the number of increments of counter~$i$ is equal to the number of decrements of counter~$i$ in $a_0 \cdots a_{n}$ (i.e., the counter is zero) but the next line number in $w$ is not the target of the if-branch.
            Or, we have  
            \[
            \sum\nolimits_{j  \colon \instr_j = \inc{i}} \occ{a_0 \cdots a_{n}}{j} \neq 
            \sum\nolimits_{j \colon \instr_j = \dec{i}} \occ{a_0 \cdots a_{n}}{j} , 
            \]
            and $a_{n+1} \neq \ell'$,
            i.e., the number of increments of counter~$i$ is not equal to the number of decrements of counter~$i$ in $a_0 \cdots a_{n}$ (i.e., the counter is nonzero) but the next line number in $w$ is not the target of the else-branch.

    \end{enumerate}
Note that the definition of error (at position~$n$) refers to the number of increments and decrements in the prefix~$a_0 \cdots a_{n}$, which does not need to be error-free itself. 
However, if a sequence of line numbers does not have an error, then the guarded-decrement property yields the following result.

\begin{lemma}
\label{lemma-simulation}
Let $w \in \set{0,1,\ldots, k-1}^+$ with $a_0 = 0$. Then, $w$ has no errors at positions~$\set{0,1,\ldots, \size{w}-2}$ if and only if $w$ is a prefix of the projection to the line numbers of the run of $\mach$.
\end{lemma}

\begin{proof}
If $w$ has no errors at positions~$\set{0,1,\ldots, \size{w}-2}$, then an induction over $n\in\nats$ shows that $(a_n, c_0^n, c_1^n)$ with
\[c_i^n =\sum\nolimits_{j  \colon \instr_{j} = \inc{i}} \occ{a_0 \cdots a_{n-1}}{j} -
            \sum\nolimits_{j \colon \instr_{j} = \dec{i}} \occ{a_0 \cdots a_{n-1}}{j} \]
is the $n$-th configuration of the run of $\mach$.

On the other hand, projecting a prefix of the run of $\mach$ to the line numbers yields a word~$w$ without errors at positions~$\set{0,1,\ldots, \size{w}-2}$.
\end{proof}

The existence of an error can be captured by a Parikh automaton, leading to the undecidability of  the safe word problem for Parikh automata.
Let $(\aut, C)$ be a \pa accepting finite words over $\Sigma$. 
A \emph{safe} word of $(\aut, C)$ is an infinite word over $\Sigma$ such that each of its prefixes is in $L(\aut, C)$. 

Guha et al.\ have shown that the existence of a safe word for \pa is undecidable. As we rely on properties of the proof to prove further undecidability results, we give a proof sketch presenting all required details. For the omitted arguments, we refer to~\cite{GJLZ22}.

\begin{lemma}[{\protect{\cite[Lemma~21]{GJLZ22}}}]
\label{lemma_safeword}
The following problem is undecidable: Given a deterministic \pa, does it have a safe word? 
\end{lemma}

\begin{proof}[Proof Sketch]
The proof proceeds by a reduction from the nontermination problem for decrement-guarded two-counter machines.
Given such a machine~$\mach = (0:  \instr_0) \cdots (k-2:  \instr_{k-2})(k-1:  \stopp)$ let $\Sigma = \set{0, \ldots, k-1}$ be the set of its line numbers.
One can construct a deterministic \pa~$(\aut_\mach, C_\mach)$ that accepts a word~$w \in \Sigma^*$ if and only if $w = \epsilon$, $w = 0$, or if $\size{w} \geqslant 2$ and $w$ does not contain an error at position~$\size{w}-2$ (but might contain errors at earlier positions).
Intuitively, the automaton checks whether the second-to-last instruction is executed properly. 
The following is then a direct consequence of Lemma~\ref{lemma-simulation}: $(\aut_\mach, C_\mach)$ has a safe word if and only if $\mach$ does not terminate.

Intuitively, the deterministic \pa~$(\aut_\mach, C_\mach)$ keeps track of the occurrence of line numbers with increment and decrement instructions of each counter (using four dimensions) and two auxiliary dimensions to ensure that the two cases in Condition~\ref{caseif} of the error definition on Page~\pageref{caseif} are only checked when the second-to-last letter corresponds to a goto instruction. All further details of the construction can be found in~\cite{GJLZ22}.
\end{proof}

After these preparations, we start with the universality problem.

\begin{theorem}
\label{thm_univundec}
The following problem is undecidable: Given an \hdpa~$(\aut, C)$ over $\Sigma$, is $L(\aut, C) = \Sigma^*$?
\end{theorem}

\begin{proof}
By reduction from the nontermination problem for decrement-guarded Minsky machines:
Given such a machine~$\mach$ with line numbers~$0,1,\ldots, k-1$ where $k-1$ is the stopping instruction, fix $\Sigma = \set{0,1,\ldots,k-1}$ and consider $L_\mach = L_\mach^0 \cup L_\mach^1$ with
\begin{align*}
L_\mach^0 {}=&{} \set{ w=a_0 \cdots a_m \in \Sigma^* \mid  \text{ $a_0 \neq 0$  or $\occ{w}{k-1} = 0$}} \text{ and }\\ 
L_\mach^1 {}=&{}\set{w=a_0 \cdots a_m \in \Sigma^* \mid \text{$a_0 = 0$, $\occ{w}{k-1} \geqslant 1$, and $w$ contains an error before the first $k-1$}}    
\end{align*}
We claim that $\mach$ does not terminate if and only if $L_\mach$ is universal.

First, assume $\mach$ does terminate. Then, projecting the terminating run to its line numbers yields a sequence~$w \in \Sigma^*$ starting with $0$, ending with $k-1$, and with no errors. This word is not in $L_\mach$, so $L_\mach$ is not universal.

Now, assume $\mach$ does not terminate and consider some $w \in \Sigma^*$. If $w$ does not start with $0$ or contains no $k-1$, then it is in $L_\mach$. Thus, now consider the case where $w$ starts with $0$ and contains a $k-1$. Towards a contradiction, assume that the sequence up to the first $k-1$ does not contain an error.
Then, Lemma~\ref{lemma-simulation} implies that $\mach$ terminates, i.e., we have derived the desired contradiction. 
Hence, $w$ contains an error before the first $k-1$, i.e., $w$ is in $L_\mach$. 
So, $L_\mach$ is indeed universal.

So, it remains to show that $L_\mach$ is accepted by an \hdpa which can be effectively constructed from $\mach$.
As \hdpa are closed under union and generalize finite automata, we only have to consider $L_\mach^1$.  In the proof of Lemma~\ref{lemma_safeword}, we have constructed a \dpa accepting a word if it does not have an error at the second-to-last position.
Thus, using complementation and intersection with a regular language, we obtain a \dpa that accepts a word if there is an error at the second-to-last position.

But the previous automaton only checks whether the error occurs at the second-to-last position.
To find an error at any position, we use history-determinism to guess the prefix with the error and then stop updating the counters and the state, so that their current values can be checked at the end of the run (cf.\ Example~\ref{example_nondyck}).
This automaton can easily be turned into one that additionally checks that the first letter is a $0$ and that there is at least one $k-1$ in the input, but not before the error.
Thus, an \hdpa accepting $L_\mach^1$ can be effectively constructed from $\mach$. 
\end{proof}

Note that there is no \dpa accepting $L_\mach$, as universality is decidable for \dpa.
Thus, the guessing described in the proof above is not avoidable.

The next results follow more or less immediately from the undecidability of universality.

\begin{theorem}
\label{thm_furtherundecidability}
The following problems are undecidable:
\begin{enumerate}
        \item Given two \hdpa~$(\aut_0, C_0)$ and $(\aut_1, C_1)$, is $L(\aut_0, C_0) \subseteq L(\aut_1, C_1)$?
        \item Given two \hdpa~$(\aut_0, C_0)$ and $(\aut_1, C_1)$, is $L(\aut_0, C_0) = L(\aut_1, C_1)$?
        \item Given an \hdpa~$(\aut, C)$, is $L(\aut, C)$ regular?
        \item Given an \hdpa~$(\aut, C)$, is $L(\aut, C)$ context-free?
\end{enumerate}
\end{theorem}

\begin{proof}
Undecidability of inclusion and equivalence follows immediately from Theorem~\ref{thm_univundec} (and even holds if the input~$(\aut_0, C_0)$ is replaced by a \dfa accepting $\Sigma^*$), so let us consider the regularity problem.

Let $L_\mach' = L_\mach^0 \cup  L_\mach^1 \cup L_\mach^2$ where $L_\mach^0$ and $L_\mach^1$ are defined as in the proof of Theorem~\ref{thm_univundec} and with
\begin{align*}
L_\mach^2 = \{w=a_0 \cdots a_m \in \Sigma^* \mid&{} \text{ $a_0 = 0$ and $\occ{w}{k-1} \geqslant 1$ and the suffix of $w$ after the first occurrence} \\ 
&\text{of the letter~$k-1$ is \emph{not} of the form~$0^n1^n$ for some $n > 0$}\}.
\end{align*}
Note that we assume without loss of generality $k - 1 \geqslant1$, i.e., $\mach$ has at least one non-stopping instruction. 

A \dpa for $L_\mach^2$ is straightforward to construct.
Hence, given $\mach$, one can effectively construct an \hdpa for $L_\mach'$ as \hdpa are closed under union.
Now, we claim that $L_\mach'$ is regular if and only if $\mach$ does not terminate. 

First, assume that $\mach$ does not terminate. 
Then as shown in the proof of Theorem~\ref{thm_univundec}, $L_\mach \subseteq L_\mach'$ is equal to $\Sigma^*$. 
Hence, $L_\mach' = \Sigma^*$ as well, which is regular.

Now, assume that $\mach$  terminates and let $w_t \in \Sigma^*$ be the projection of the run of $\mach$ to its line numbers.
Note that $w_t$ starts with $0$, ends with $k-1$, and does not contain an error.
We show that $L_\mach' = \Sigma^* \setminus w_t \cdot \set{0^n1^n \mid n > 0}$, which is not regular.

First, if $w$ is in $w_t \cdot \set{0^n1^n \mid n > 0}$, then it is not in $L_\mach'$.
Now, if $w$ is not in $L_\mach'$, then it has to start with $0$ and has to contain a $k-1$ (otherwise, $w$ would be in $L_\mach^0 \subseteq L_\mach'$).
Furthermore, $w$ cannot contain an error before the first $k-1$ (otherwise, $w$ would be in $L_\mach^1 \subseteq L_\mach'$).
Thus, $w_t$ is a prefix of $w$.
Finally, the suffix of $w$ after the first $k-1$ has to be of the form $0^n1^n$ for some $n > 0$ (otherwise, $w$ would be in $L_\mach^2 \subseteq L_\mach'$). 
Altogether, $w$ is in $w_t \cdot \set{0^n1^n \mid n > 0}$.

Finally, using $0^n1^n2^n$ instead of $0^n1^n$ in the definition of $L_\mach^2$ in the previous proof, which is not context-free, allows us to show that the context-freeness problem is also undecidable: Given an \hdpa~$(\aut, C)$, is $L(\aut, C)$ context-free?
\end{proof}

Next, let us introduce the model-checking problem (for safety properties): 
A transition system~$\tsys = (V, v_\init, E, \lambda)$ consists of a finite set~$V$ of vertices containing the initial state~$v_\init \in V$, a transition relation~$E \subseteq V \times V$, and a labeling function~$\lambda \colon V\rightarrow \Sigma$ for some alphabet~$\Sigma$.
A (finite and initial) path in $\tsys$ is a sequence~$v_0v_1 \cdots v_n \in V^+$ such that $v_0 = v_\init$ and $(v_i,v_{i+1}) \in E$ for all $0 \leqslant i <n$.
Infinite (initial) paths are defined analogously.
The trace  of a path~$v_0v_1 \cdots v_n$ is $\lambda(v_0)\lambda(v_1) \cdots \lambda(v_n) \in \Sigma^+$.
We denote the set of traces of paths of $\tsys$ by $\traces(\tsys)$.

The model-checking problem for \hdpa asks, given an \hdpa~$\aut$ and a transition system~$\tsys$, whether $\traces(\tsys) \cap L(\aut) = \emptyset$?
Note that the automaton specifies the set of \emph{bad prefixes}, i.e., $\tsys$ satisfies the specification encoded by $\aut$ if no trace of $\tsys$ is in $L(\aut)$.

As the model-checking problem for \pa is decidable, so is the model-checking problem for \hdpa, which follows from the fact that a transition system~$\tsys$ can be turned into an \nfa and hence into a \pa~$\aut_\tsys$ with $L(\aut_\tsys) = \traces(\tsys)$. 
Then, closure under intersection and decidability of nonemptiness yields the desired result.

\begin{theorem}
\label{thm_mc}
The model-checking problem for \hdpa is decidable.
\end{theorem}

Let us conclude by mentioning that the dual problem, i.e., given a transition system~$\tsys$ and an \hdpa~$\aut$, does every infinite path of $\tsys$ have a prefix whose trace is in $L(\aut)$, is undecidable. 
This follows from recent results on Parikh automata over infinite words~\cite[Theorem~17]{GJLZ22}, i.e., that model checking for Parikh automata with reachability conditions is undecidable. Such automata are syntactically equal to Parikh automata over finite words and an (infinite) run is accepting if it has a prefix ending in an accepting state whose extended Parikh image is in the semilinear set of the automaton. 

Table~\ref{tab:dec} compares the decidability of standard problems for \hdpa with those of \dpa, \uca, and \pa. 

\begingroup

\begin{table}[h]
    \caption{Decision problems for history-deterministic Parikh automata (in grey) and comparison to other types of Parikh automata (results are from~\cite{CFM12,CFM13,KR}).}
    \label{tab:dec}
    \centering
\begin{tabular}{lccccccc}
\toprule
 & Nonemptiness & Finiteness &  Universality & Inclusion & Equivalence &  Regularity & Model Checking\\
 % & $\neq\emptyset$? & finite? &  $=\Sigma^*$? & $\subseteq$? & $=$? &  regular? & MC\\
 \midrule
\dpa    & Y & Y & Y & Y & Y & Y & Y \\
\rowcolor{lightgray!50}
\hdpa   & Y & Y & N & N & N & N & Y \\
\uca    & Y & Y & Y & Y & Y & Y & Y \\
\pa     & Y & Y & N & N & N & N & Y \\

\bottomrule
\end{tabular}
\end{table}
\endgroup

Finally, we consider the problems of deciding whether a Parikh automaton is history-deterministic and whether it is equivalent to some \hdpa. 
Both of our proofs follow arguments developed for similar results for history-deterministic pushdown automata~\cite[Theorem~6.1]{LZ22}.

\begin{theorem}
\label{thm_hdnessundec}
The following problems are undecidable:
\begin{enumerate}
    \item\label{dec-hd} Given a \pa~$(\aut, C)$, is it history-deterministic?

    \item\label{dec-hdeq} Given a \pa~$(\aut, C)$, is it equivalent to some \hdpa?
\end{enumerate}
\end{theorem}

\begin{proof}
\ref{dec-hd}.) 
We say that a \pa~$(\aut, C)$ is \emph{length-complete} if for every $n \in \nats$, there is some word of length~$n$ in $L(\aut, C)$. 
Furthermore, let $h \colon \Sigma^* \rightarrow \set{\#}^*$ be the homomorphism induced by mapping each $a \in \Sigma$ to $\#$.
If $(\aut, C)$ is length-complete, then $\set{h(w) \mid w \in L(\aut, C) } = \set{\#}^*$.
Note that the \dpa's~$(\aut_\mach, C_\mach)$ we have constructed in the proof of Lemma~\ref{lemma_safeword} are length-complete.
Thus, we have actually shown that the following problem is undecidable: Given a length-complete \dpa~$(\aut, C)$, does it have a safe word?

We reduce from the safe word problem for length-complete \dpa.
Given such a \dpa~$(\aut, C)$ with $L(\aut, C) \subseteq \Sigma^*$, let $(\aut', C)$ be the \pa obtained from $(\aut, C)$ by replacing each letter~$a \in \Sigma$ by $\#$.
Note that $(\aut', C)$ accepts $\set{h(w) \mid w \in L(\aut, C) }$, which is equal to $\set{\#}^*$ due to the length-completeness of $(\aut, C)$.
We show that $(\aut', C)$ is history-deterministic if and only if $(\aut, C)$ has a safe word, which completes our proof. 

So, let $(\aut', C)$ be history-deterministic, i.e., it has a resolver~$\hstrat \colon \set{\#}^* \rightarrow \Delta_{\aut'}$, where $\Delta_{\aut'}$ is the set of transitions of $\aut'$. 
Note that $\hstratiter(\#^n)$ is a prefix of $\hstratiter(\#^{n'})$ whenever $n' \geqslant n$.
Hence, there is a unique infinite sequence
\[(q_0, (\#, \vec{v}_0), q_1) (q_1, (\#, \vec{v}_1), q_2)(q_2, (\#, \vec{v}_2), q_3) \cdots \] such that $\hstratiter(\#^n) = (q_0, (\#, \vec{v}_0), q_1)\cdots (q_{n-1}, (\#, \vec{v}_{n-1}), q_n)$.
Each of the $\hstratiter(\#^n)$ is an accepting run of $\aut'$ processing $\#^n$, as $\#^n$ is accepted by $(\aut',C)$ and $\hstrat$ is a resolver.

Now, as each transition of $\aut'$ is obtained from a transition of $\aut$, there is some infinite word~$a_0 a_1 a_2 \cdots $ over $\Sigma$ such that $(q_0, (a_0, \vec{v}_0), q_1)\cdots (q_{n-1}, (a_{n-1}, \vec{v}_{n-1}), q_n)$ is an accepting run of $\aut$ for every $n \geqslant 1$. 
Thus, we have $a_0 \cdots a_{n-1} \in L(\aut, C)$ for every $n \geqslant 1$. 
Lastly, we have $\epsilon \in L(\aut, C)$, as $\epsilon \in L(\aut', C)$.
Altogether, $a_0 a_1 a_2 \cdots$ is indeed a safe word for $(\aut, C)$.

Now, assume $(\aut, C)$ has a safe word, say $a_0a_1a_2 \cdots$.
Thus, $\epsilon$ is in $L(\aut, C)$ which implies that the initial state of $\aut$ is accepting and the zero vector is in $C$.
Further, as the automaton~$(\aut, C)$ is deterministic, there is a sequence~$(q_0, (a_0, \vec{v}_0), q_1) (q_1, (a_1, \vec{v}_1), q_2)(q_2, (a_2, \vec{v}_2), q_3) \cdots $
such that $(q_0, (a_0, \vec{v}_0), q_1)\cdots (q_{n-1}, (a_{n-1}, \vec{v}_{n-1}), q_n)$ is an accepting run of $\aut$ for all $n\in\nats$. 

By construction of $(\aut', C)$, $(q_0, (\#, \vec{v}_0), q_1) \cdots (q_{n-1}, (\#, \vec{v}_{n-1}), q_n)$ is an accepting run of $\aut'$ on $\#^n$, for each $n \geqslant 1$.
Now, we define $\hstrat(\#^n) = (q_{n-1}, (\#, \vec{v}_{n-1}), q_n)$ for each $n \geqslant 1$. 
Then, $\hstratiter(\#^n) = (q_0, (\#, \vec{v}_0), q_1) \cdots (q_{n-1}, (\#, \vec{v}_{n-1}), q_n)$, i.e., it is an accepting run of $(\aut', C)$ processing $\#^n$. 
Hence, $\hstrat$ is a resolver for $(\aut', C)$, i.e., $(\aut', C)$ is history-deterministic.

\ref{dec-hdeq}.)
We reduce from the universality problem for \pa, which is undecidable~\cite[Property~7]{KR}, via a series of automata transformations. 

First, given a \pa~$(\aut, C)$ processing words over~$\Sigma$, we construct a \pa~$(\aut', C')$ with
\[
L(\aut', C')  = L(\aut, C) \cup L(\aut, C)\cdot \# \cdot (\Sigma_\#)^*
\]
where $\Sigma_\# = \Sigma\cup\set{\#}$ for some fresh~$\# \notin \Sigma$. 
Note that $(\aut, C)$ is universal, i.e., $L(\aut, C) = \Sigma^*$, if and only if $(\aut', C')$ is universal, i.e., $L(\aut', C') = (\Sigma_\#)^*$. 

Recall that the language~$E \subseteq \set{a,b}^*$ from the proof of Theorem~\ref{thm_sep} is accepted by some \pa, but not by any \hdpa.
Thus, we can construct a \pa~$(\aut_\vee, C_\vee)$ with 
\[
L(\aut_\vee, C_\vee) = \set{ w \in (\Sigma_\#\times \set{a,b})^* \mid \pi_1(w) \in L(\aut', C') \text{ or } \pi_2(w) \in E }.
\]
Here, $\pi_i$ is the projection to the $i$-th component, i.e., the homomorphism defined by $\pi_i(a_1, a_2) = a_i$.
We show that $L(\aut, C)$ is universal if and only if $L(\aut_\vee,C_\vee)$ is accepted by some \hdpa.

First, assume that $L(\aut, C)$ is universal. Then, $L(\aut', C')$ and $L(\aut_\vee, C_\vee)$ are universal as well. Hence, $L(\aut_\vee, C_\vee)$ is accepted by some DFA, and therefore also by an \hdpa.

Now, assume that $L(\aut, C)$ is not universal, i.e., there is some $u = a_0 \cdots a_{n-1} \notin L(\aut, C)$.
Towards a contradiction, we assume there is an \hdpa~$(\aut_a, C_a)$ with $L(\aut_a, C_a) = L(\aut_\vee,C_\vee)$, say with resolver~$\hstrat_a$.
We show that the language
\[P =  \left\{ w \in \set{a,b}^{\geqslant \size{u}} \:\middle|\: \binom{u\#^{\size{w} - \size{u}}}{w} \in L(\aut_\vee, C_\vee) \right\} \]
is also accepted by some \hdpa~$(\aut_P, C_a)$.
Note that $P$ contains exactly the words~$w \in\set{a,b}^{\geqslant\size{u}} \cap \equal$, as $u\#^{\size{w} - \size{u}}$ is not in $L(\aut',C')$ for any $w$.
The fact that $P$ is accepted by some \hdpa yields the desired contradiction:
We have $P \subseteq \equal$ and $\equal \setminus P$ is finite.
Hence, $\equal \setminus P$ is accepted by some \dfa and therefore also by some \hdpa.
Thus, due to closure of \hdpa under union, if $P$ is accepted by some \hdpa, so is $\equal$.
This contradicts that $\equal$ is not accepted by any \hdpa (see Theorem~\ref{thm_sep}).

So, let $\aut_a = (Q, (\Sigma_\#\times \set{a,b})\times D, q_\init, \Delta, F)$. 
We define $\aut_P = (Q \times \set{0, 1, \dots, \size{u}}, \set{a,b}\times D, (q_\init, 0), \Delta', F \times \set{\size{u}})$ with
\begin{align*}
\Delta' ={}&{} \left\{ 
((q, j), (a,\vec{v}) , (q', j+1)) \:\middle|\: (q, (\binom{a_j}{a},\vec{v}), q') \in \Delta \text{ and } j<\size{u} \right\} \,\cup\\
{}&{}\left\{((q, \size{u}), (a,\vec{v}) , (q', \size{u})) \:\middle|\: (q, (\binom{\#}{a},\vec{v}), q') \in \Delta \right\}
\end{align*}
Intuitively, to obtain $\aut_P$, we hardcode $u$ into the state space in order to restrict the runs of $\aut_a$ to those that process words of the form~$u\#^*$ in the first component (which is projected away). 
Hence, $(\aut_P, C_a)$ does indeed accept $P$.

Now it remains to observe that $(\aut_P, C_a)$ has a resolver:
Turning $(\aut_a, C_a)$ into $(\aut_P, C_a)$ does not introduce nondeterminism, i.e., a resolver~$\hstrat$ for $(\aut_a,C_a)$ can easily be turned into one for $(\aut_P,C_a)$.
So, $(\aut_P,C_a)$ is an \hdpa.
\end{proof}

%%%%%%%%%%%%%%%%%%%%%%%%%%%%%%%%%%%%%%%%%%%%%%%%%%%%%%%%%%%%%%%
%%%%%%%%%%%%%%%%%%%%%%%%%%%%%%%%%%%%%%%%%%%%%%%%%%%%%%%%%%%%%%%
%%%%%%%%%%%%%%%%%%%%%%%%%%%%%%%%%%%%%%%%%%%%%%%%%%%%%%%%%%%%%%%
\section{Resolvers}
\label{sec_resolvers}

So far, a resolver for an \hdpa is a function~$\hstrat\colon \Sigma^+ \rightarrow \Delta$.
This definition does not put any restrictions on the complexity of the function~$\hstrat$. 
But how complex is it to resolve the nondeterminism in an \hdpa?
For history-deterministic finite automata, finite automata with output suffice to resolve their nondeterminism~\cite[Lemma~1]{BKKS13} while pushdown automata with output are not sufficient to resolve the nondeterminism in history-deterministic pushdown automata~\cite[Theorem~7.3]{LZ22}\cite[Lemma~10.2]{GJLZ24}.

It is straightforward to show that every \hdpa has a positional resolver (i.e., one whose decision is only based on the last state of the run constructed thus far and on the extended Parikh image induced by this run) and that \hdpa that have finite-state resolvers (say, implemented by a Mealy machine) can be determinized by taking the product of the \hdpa and the Mealy machine.
In fact, both proofs are simple adaptions of the corresponding ones for history-deterministic pushdown automata~\cite[Lemma 10.1 and Remark~10.3]{GJLZ24}. 
Hence, we focus on the question whether the nondeterminism in an \hdpa can be resolved by functions implemented by more expressive formalisms, i.e., we are looking for deterministic machines that implement resolvers for \hdpa. 
A natural question is then whether \dpa are sufficient for this task, i.e., the question whether the nondeterminism in an \hdpa can be resolved by \dpa. 

Let us formalize this question.
Fix a \pa~$(\aut, C)$ with $\aut = (Q, \Sigma \times D, q_\init, \Delta, F)$ throughout this section. 
We assume (without loss of generality) that $(\aut, C)$ is complete.
Recall that a function~$\hstrat\colon \Sigma^+ \rightarrow \Delta$ is a resolver for $(\aut, C)$ if, for every~$w \in L(\aut, C)$, we have that $\hstratiter(w)$ is an accepting run of $(\aut, C)$. 
Here, $\hstratiter \colon \Sigma^* \rightarrow \Delta^*$ is defined as $\hstratiter(\epsilon) = \epsilon$ and $\hstratiter(a_0 \cdots a_n) = \hstratiter(a_0 \cdots a_{n-1})\cdot \hstrat(a_0 \cdots a_n)$. 
In the following, we call such a resolver a letter-based resolver.

Here, it is more convenient to work with resolvers that not only get the prefix~$w$ to be processed as input, but the run prefix constructed thus far (from which $w$ can be reconstructed). In general, one can, given an input word, reconstruct the run prefix induced by the resolver (see the proof of Lemma~\ref{lemma_resolverdefequivalence} below). 
However, this reconstruction cannot necessarily be implemented by an automaton (see Item~\ref{example_nondyckcont_nonaut} of Example~\ref{example_nondyckcont} below).
Thus, while these two definitions are equivalent in general, this is not necessarily true when restricted to automata-implementable resolvers.

Let $\runs{(\aut,C)} \subseteq \Delta^*$ denote the set of run prefixes of $(\aut,C)$ starting in the initial state. 
We now consider functions of the form~$\hstratprime \colon \runs{(\aut,C)} \times \Sigma \rightarrow \Delta$ and define their iteration~$\hstratiterprime \colon \Sigma^* \rightarrow \Delta^*$ inductively as $\hstratiterprime(\epsilon) = \epsilon$ and \[\hstratiterprime(a_0 \cdots a_n) = \hstratiterprime(a_0 \cdots a_{n-1}) \cdot \hstratprime(\hstratiterprime(a_0 \cdots a_{n-1}), a_n).\]
We say that $\hstratprime$ is a transition-based resolver if, for every~$w \in L(\aut, C)$, $\hstratiterprime(w)$ is an accepting run of $(\aut, C)$. 

\begin{lemma}
\label{lemma_resolverdefequivalence}
A \pa~$(\aut, C)$ has a letter-based resolver if and only if it has a transition-based resolver.
\end{lemma}

\begin{proof}
Let $\hstrat\colon \Sigma^+ \rightarrow \Delta$ be a letter-based resolver for $(\aut, C)$. 
We inductively define the function~$\hstratprime\colon \runs{(\aut,C)} \times \Sigma \rightarrow \Delta$ via $\hstratprime(\epsilon, a) = \hstrat(a)$ and $\hstratprime(\rho, a) = \hstrat(wa)$ where $w \in \Sigma^+$ is the word labelling the sequence~$\rho \in \Delta^+$.
Then, an induction shows that he have $\hstratiterprime = \hstratiter$, which implies that $\hstratprime$ is a transition-based resolver.

For the other implication, let $\hstratprime \colon \runs{(\aut,C)} \times \Sigma \rightarrow \Delta$ be a transition-based resolver. 
We inductively define the function~$\hstrat \colon \Sigma^+ \rightarrow \Delta$  via $\hstrat(a) = \hstratprime(\epsilon, a)$ and $\hstrat(a_0 \cdots a_n) = \hstratprime(\hstratiterprime(a_0 \cdots a_{n-1}), a_n)$. 
Again, an induction shows that he have $\hstratiter = \hstratiterprime$, which implies that $\hstrat$ is a letter-based resolver.
\end{proof}

Next, we define how to implement a transition-based resolver via automata.
Unlike previous work, which employed transducers (automata with outputs) to implement resolvers, we prefer here to use a different approach and show that for each transition~$\delta$, the inverse image of $\delta$ is accepted by an automaton. 
Taking the product of these automata for each $\delta$ yields a transducer implementing a resolver, but not taking the product simplifies our constructions below considerably.

Given a transition-based resolver~$\hstratprime \colon \runs{(\aut,C)} \times \Sigma \rightarrow \Delta$, $a \in \Sigma$, and $\delta \in \Delta$, we define
\[\hstratprimeinv{a}{\delta} = \set{\rho \in \runs{(\aut,C)} \mid \hstratprime(\rho, a ) = \delta},\]
i.e., the set of run prefixes for which, when the next input letter to be processed is an $a$, $\hstratprime$ extends the run prefix with the transition~$\delta$. 
Note that $\hstratprime$ is completely specified by the collection of sets~$\hstratprimeinv{a}{\delta}$.
Conversely, if we have a collection of sets~$R_{a,\delta}$ that partitions $\runs{(\aut,C)}$, then it induces a function mapping finite runs and letters to transitions, i.e., it has the same signature as a transition-based resolver.
We say that a transition-based resolver~$\hstratprime$ is \dpa-implementable, if every $\hstratprimeinv{a}{\delta} $ is accepted by a \dpa.

\begin{example}
\label{example_nondyckcont}
Consider the \hdpa~$(\aut, C)$ from Example~\ref{example_nondyck}.
For $b \in \set{0,1}$, let $\delta_{b, \ell}$ be the self-loop processing $b$ on the left state~$q_c$, let $\delta_{b, m}$ be the transition processing $b$ leading from~$q_c$ to $q_n$, and let $\delta_{b, r}$ be the self-loop processing $b$ on the right state~$q_n$.

\begin{enumerate}
    \item The following languages are all accepted by a \dpa, for $b \in \set{0,1}$:
\
\begin{itemize}
    \item $R_{b, \delta_{b,\ell}} = \set{\delta_{b_0 ,\ell} \cdots \delta_{b_{n-1} ,\ell} \mid n \geqslant 0 \text { and } b_0 \cdots b_{n-1}b \text{ is not non-Dyck} }$.
    \item $R_{b, \delta_{b,m}} = \set{\delta_{b_0 ,\ell} \cdots \delta_{b_{n-1} ,\ell} \mid n > 0 \text { and } b_0 \cdots b_{n-1}b \text{ is non-Dyck} }$.
    \item $R_{b, \delta_{b,m}} = \set{\delta_{b_0 ,\ell} \cdots \delta_{b_{n-1} ,\ell}  \delta_{b_n ,m}   \delta_{b_{n+1} ,r} \cdots \delta_{b_{n+n'} ,r} \mid n > 0 \text { and } n' \geqslant 0 }$.
\end{itemize}
These languages are pairwise disjoint and their union is the set of all run prefixes.
Now, one can verify that the induced function is a transition-based resolver for $(\aut, C)$, i.e., $(\aut, C)$ has a \dpa-implementable resolver.

    \item \label{example_nondyckcont_nonaut}
    Now, let us consider implementation of letter-based resolvers for $(\aut, C)$.
    Recall that the letter-based resolver presented in Example~\ref{example_nondyck}, call it ~$\hstrat$, takes the transition to $q_n$ as soon as possible, i.e., with the first non-Dyck prefix of the input.
    One can in fact show that this is the only letter-based resolver of $(\aut, C)$, as staying in $q_c$ with a non-Dyck prefix~$w$ means that the run induced by the resolver is not accepting, a contradiction.  

    Now, assume $\hstrat$ is \dpa-implementable (in the sense that $\set{w \in \set{0,1}^+ \mid \hstrat(w) = \delta}$ is accepted by a \dpa for every transition~$\delta$). 
    Then, the language
    \[
    \set{w \in \set{0,1}^+ \mid \hstrat(w) = \delta_{0,\ell}} \cup \set{w \in \set{0,1}^+ \mid \hstrat(w) = \delta_{1,\ell}}.
    \]
    is also accepted by a $\dpa$, as \dpa are closed under union. 
    However, by definition of $\hstrat$, this is the language of words that do not have a non-Dyck prefix. 
    It is known that this language is not accepted by any \pa~\cite[Proposition~11]{CFM13}.  
    Hence, $\hstrat$ is not \dpa-implementable and, due to uniqueness of $\hstrat$, $(\aut, C)$ does not have \dpa-implementable letter-based resolvers.

    Note that $\hstrat$ is not even implementable by \pa (with the obvious definition), as the language of words without non-Dyck prefixes is not even accepted by (possibly nondeterministic) \pa~\cite[Proposition~11]{CFM13}.
\end{enumerate}
\end{example}

Note that a transition-based resolver operates on the run constructed thus far, which contains information on whether there was a non-Dyck prefix in the word processed thus far: If there is one, then the resolver continued the run on that prefix by moving to state~$q_n$. Thus, any run consistent with the resolver ends in state~$q_n$ if and only if the word processed thus far has a non-Dyck prefix. 
This saves a transition-based resolver from checking every input for non-Dyckness, a property that cannot be checked by an (even nondeterministic) \pa. 
A letter-based resolver, on the other hand, does not get the information in which state the run constructed thus far ends, and a \dpa-implementable one also cannot compute that information. Thus, \dpa-implementable letter-based resolvers are strictly weaker than \dpa-implementable transition-based resolvers.

\begin{remark}
All \hdpa we present in examples and proofs in this work have a \dpa-implementable transition-based resolver.
\end{remark}

Due to Example~\ref{example_nondyckcont}, we will consider only transition-based resolver in the remainder of this section and, for the sake of brevity, call them resolvers in the following.

Now, we investigate whether every \hdpa has a \dpa-implementable resolver. 
Unfortunately, we only obtain a conditional result.
To explain the condition, we need to introduce some background on the game-theoretic characterization of history-determinism due to Bagnol and Kuperberg~\cite{BK18}.

Recall that we fixed a complete \pa~$(\aut, C)$ with $\aut = (Q, \Sigma \times D, q_\init, \Delta, F)$.
We define the one-token game on $(\aut,C)$, a perfect-information zero-sum game played between two players called \attacker and \defender. The positions of the game consist of pairs~$(q,q') \in Q \times Q$. The game is played in rounds~$i = 0,1,2,\ldots$ starting from the initial position~$(q_0, q_0') = (q_\init, q_\init)$. In round~$i$ at position~$(q_i, q_i')$:
\begin{itemize}
    \item First \attacker picks a letter~$a_i \in \Sigma$.
    \item Then, \defender picks a transition from $\Delta$ of the form~$(q_i, (a_i, \vec{v}_i), q_{i+1})$ for some $\vec{v}_i \in D$ and some $q_{i+1} \in Q$. 
    \item Then, \attacker picks a transition from $\Delta$ of the form~$(q_i', (a_i, \vec{v}_i'), q_{i+1}')$ for some $\vec{v}_i' \in D$ and some $q_{i+1}' \in Q$. 
    \item Then, round~$i$ ends and the play proceeds from position~$(q_{i+1}, q_{i+1}')$ in round~$i+1$.
\end{itemize}
Due to completeness of $(\aut, C)$, this is well-defined, i.e., each player has at least one move available. Thus, plays are infinite and induce two infinite sequences
\[(q_0, (a_0, \vec{v}_0), q_{1})(q_1, (a_1, \vec{v}_1), q_{2})(q_2, (a_2, \vec{v}_2), q_{3}) \cdots \text{\quad and \quad} (q_0', (a_0, \vec{v}_0'), q_{1}')(q_1', (a_1, \vec{v}_1'), q_{2}')(q_2', (a_2, \vec{v}_2'), q_{3}') \cdots\] 
picked by \defender and \attacker, respectively.
Each prefix of the sequences is, by definition of the game, a run of $(\aut, C)$.
\defender wins the play if for every $n \geqslant0$ for which the run~$(q_0', (a_0, \vec{v}_0'), q_{1}') \cdots (q_{n-1}', (a_{n-1}, \vec{v}_{n-1}'), q_{n}')$ picked by \attacker is accepting, the run~$(q_0, (a_0, \vec{v}_0), q_{1}) \cdots (q_{n-1}, (a_{n-1}, \vec{v}_{n-1}), q_{n})$ picked by \defender is also accepting.
A strategy for \defender maps a finite sequence of transitions (picked by \attacker) and a letter to a transition. 
The notions of winning strategies and \defender winning the one-token game are defined as expected.

\begin{lemma}
\label{lemma_onetokencorrectness}
\defender wins the one-token game on a \pa~$(\aut, C)$ if and only if $(\aut, C)$ is an \hdpa.
\end{lemma}

\begin{proof}
The proof here is analogous to the one for finite automata~\cite[Lemma~7 (case of $k=1$)]{LR13}, but we need to repeat it here as our further results depend on properties of the one-token game that are shown in this proof.

One direction is straightforward: a transition-based resolver can directly be turned into a winning strategy for \defender which ensures that whenever the word picked by \attacker thus far is in $L(\aut, C)$, then the finite run constructed thus far by \defender is accepting (this is what a resolver does after all). In particular, when the run constructed by \attacker is accepting (which implies that the word picked thus far is in $L(\aut, C)$), then the run constructed by \defender is also accepting. Thus, \defender wins the one-token game.

For the converse direction, we need to show that a winning strategy for \defender can be turned into a transition-based resolver for $(\aut, C)$. 
Throughout the proof, we fix such a winning strategy~$\sigma$.

We begin by introducing non-residual transitions.
Assume \attacker has picked a word~$w$ thus far during a play and the players have picked sequences~$\rho$ (\defender) and $\rho'$ (\attacker) of transitions (processing both $w$ by definition).
Now, let \attacker pick the letter~$a$ to begin the next round and let $\delta$ be a transition that \defender could pick to continue the round.
We say that $\delta$ is non-residual w.r.t.\ $\rho$ and $a$ (but independent of $\rho'$) if there is a $w' \in \Sigma^*$ such that
\begin{itemize}
    
    \item there is an accepting run~$\rho^*$ with prefix~$\rho$ such that $\rho^*$ processes $waw'$, 
    
    \item but there is no accepting run with prefix~$\rho\delta$ that processes $waw'$.

\end{itemize}
A transition is residual, if it is not non-residual.

If $\rho = \rho'$ in the situation above, then $\sigma$ must pick a residual transition. If not, then \attacker wins by picking the letters of $w'$ and the transitions completing the accepting run~$\rho^*$ processing $waw'$. By definition, \defender cannot extend $\rho\delta$ to an accepting run processing $waw'$. 
This contradicts $\sigma$ being winning, as such a play is not winning for \defender.

We define the function~$\hstratprime \colon \runs{(\aut,C)} \times \Sigma \rightarrow \Delta$. Let $\rho \in \runs{(\aut,C)}$ be a run prefix of $(\aut,C)$ processing the word~$w$, let $a \in \Sigma$, and let $\delta$ be the transition picked by $\sigma$ when \attacker has picked $w$ thus far, both players have picked $\rho$, and \attacker has picked $a$ to start the next round.
Then, we define $\hstratprime(\rho, a) = \delta$.
Note that $\delta$ must be residual w.r.t.\ $\rho$ and $a$, as argued above.

Now, an induction shows that if $w$ is in $L(\aut, C)$, then the run induced by $\hstratprime$ over every prefix of $w'$ of $w$ can be extended to an accepting run that processes $w$, as $\hstratprime$ only yields residual transitions. Thus, for $w' = w$ we obtain that the run induced by $\hstratprime$ is accepting, i.e., $\hstratprime$ is a resolver. 
\end{proof}

Given a \pa~$(\aut, C)$ with $\aut = (Q, \Sigma \times D, q_\init, \Delta, F)$ and $C \subseteq \nats^d$, the one-token game on $(\aut, C)$ can be modelled as a safety game on an infinite graph as follows (see, e.g., \cite{GTW02} for a general introduction to graph-based games):
\begin{itemize}
    
    \item The set of vertices is $( Q \times \nats^d \times Q \times \nats^d) \cup ( Q \times \nats^d \times Q \times \nats^d \times \Sigma) \cup ( Q \times \nats^d \times Q \times \nats^d \times \Sigma\times \Delta)$.
    
    \item The initial vertex is $(q_\init, \vec{0}, q_\init, \vec{0})$ where $\vec{0}$ is the $d$-dimensional zero vector.
    
    \item The vertices of \attacker are those in $( Q \times \nats^d \times Q \times \nats^d) \cup ( Q \times \nats^d \times Q \times \nats^d \times \Sigma\times \Delta)$.
    
    \item The vertices of \defender are those in $( Q \times \nats^d \times Q \times \nats^d \times \Sigma)$.
    
    \item The edge relation is defined as follows:
    
    \begin{itemize}
        
        \item A vertex~$(q,\vec{v}, q', \vec{v}')$ has the successor~$(q,\vec{v}, q', \vec{v}',a)$ for each $a \in \Sigma$, simulating that \attacker picks a letter at the start of a new round.

        \item A vertex~$(q,\vec{v}, q', \vec{v}',a)$ has the successor~$(q,\vec{v}, q', \vec{v}',a,\delta)$ for every $\delta \in \Delta$ of the form~$\delta = (q, (a, \vec{v}''), q'')$ for some $\vec{v}'' \in D$ and some $q'' \in Q$, simulating that \defender picks a transition in the second step of a round.
        
        \item A vertex~$(q,\vec{v}, q', \vec{v}',a, (q, (a, \vec{v}''), q''))$ has the successor~$(q'',\vec{v} + \vec{v}'', q''', \vec{v}' + \vec{v}''')$ for every $\delta' \in \Delta$ of the form $\delta' = (q', (a, \vec{v}'''), q''')$ for some $\vec{v}''' \in D$ and some $q''' \in Q$, simulating that \attacker picks a transition in the last step of a round. This leads to the update of the states and the counters.
        
    \end{itemize}

    \item The set of unsafe vertices is 
    \[
    \set{(q,\vec{v}, q', \vec{v}') \in Q\times \nats^d \times Q \times \nats^d \mid q' \in F \wedge \vec{v}' \in C \wedge (q \notin F \vee \vec{v} \notin C)}.
    \]
    All other vertices are safe.
\end{itemize}

As usual,  a play is an infinite sequence of vertices starting in the initial one and following the edge relation. It is winning for \defender if it contains only safe vertices.
 The winning region of \defender is the set of vertices from which \defender has a winning strategy.

\begin{remark}
\defender wins the one-token game on a \pa~$(\aut, C)$ if and only if the initial vertex of the safety game is in the winning region of \defender.
\end{remark}

The set of vertices of the safety game is isomorphic to $S \times \nats^{2d}$ for some finite set~$S$. 
We say that a subset~$A  \subseteq S \times \nats^{2d}$ is semilinear if
\[
\restr{A}{s} =  \set{\vec{v} \in \nats^{2d} \mid (s,\vec{v}) \in A}
\]
is semilinear for each $s \in S$.
Note that there are only finitely many $\restr{A}{s}$, as $S$ is finite.
Slightly abusively, we say a subset of the vertices is semilinear, if its isomorphic image is semilinear.

\begin{theorem}
\label{thm_resolversviagames}
If the winning region of \defender in the safety game induced by $(\aut, C)$ is semilinear and contains the initial vertex, then $(\aut, C)$ has a \dpa-implementable transition-based resolver.
\end{theorem}

\begin{proof}
Our construction uses residual transitions as introduced in the proof of Lemma~\ref{lemma_onetokencorrectness}. 
Let us begin with a simple observation that follows from the definition of residual transitions: if two finite runs~$\rho,\rho'$ end in the same state (the empty run ends in the initial state by convention) and have the same extended Parikh image (the empty run has the extended Parikh image~$\vec{0}$ by convention), and $a$ is a letter, then a transition is residual w.r.t.\ $\rho$ and $a$ if and only if it is residual w.r.t.\ $\rho'$ and $a$. 
Note that the safety game does not keep track of the whole play history of a play in the one-token game, but only keeps track of the last states and the extended Parikh images of the two runs constructed by the two players.
Due to the observation above, this suffices to determine residual transitions of histories while only keeping track of the last vertex of the history.

Let $W$ be the winning region of \defender in the safety game and consider a vertex of the form~$(q,\vec{v}, q,\vec{v},a) \in W$, i.e., it is \defender's turn.
Then, there is at least one successor in the winning region~$W$, which must be of the form~$(q,\vec{v}, q,\vec{v},a,\delta)$ for a residual transition~$\delta$, as argued above.
Thus, \defender can pick any residual transition to win the one-token game, again as argued above. 
Hence, we will construct a \dpa-implementable resolver that will always pick a residual transition. 

However, in general there might be several residual transitions available.
So, to define the resolver formally, we need to select one such transition.
To this end, let $\delta_0 < \delta_1 < \cdots \delta_n$ be a fixed ordering of the set~$\Delta$ of transitions of $(\aut,C)$.
We fix a letter~$a$ and a transition~$\delta_i$ and construct a \dpa that does the following, when given a finite run $\rho$:
\begin{itemize}
    \item Compute the extended Parikh image~$\vec{v}$ of $\rho$ using its counters and determine the last state~$q$ of $\rho$.
    \item If $(q,\vec{v},q,\vec{v},a,\delta_i) \in W$ and $(q,\vec{v},q,\vec{v},a,\delta_{i'}) \notin W$ for all $i' < i$, then accept. Intuitively, if $\delta_i$ is the first residual transition, then $\rho$ is accepted by the automaton associated with $\delta_i$ and $a$.
    \item If $i=0$ and $(q,\vec{v},q,\vec{v},a,\delta_{i'}) \notin W$ for all $i' \in \set{0,1,\ldots, n}$, then accept. Intuitively, if there is no residual transition for $\rho$ and $a$, then the automaton associated with $\delta_0$ and $a$ accepts. Note that this case is only used for completeness, it will never be used below as the resolver we are defining always ensures that there is at least one residual transition.
\end{itemize}
The last two conditions can be checked by the Parikh automaton we are constructing, as each $\restr{W}{(q,q,a,\delta_{i'})}$ is semilinear by assumption. Note that this requires to reflect the last state~$q$ in the extended Parikh image as done several times before.

Now, let $R_{a, \delta_i}$ be the language of the automaton we have constructed for $a$ and $\delta_i$.
The sets~$R_{a, \delta_i}$ for $a \in \Sigma$ and $i \in \set{0,1,\ldots, n}$ partition $\runs{(\aut,C)}$, and therefore induce a function~$\hstratprime\colon \runs{(\aut,C)} \times \Sigma \rightarrow \Delta$.

It remains to show that $\hstratprime$ is a transition-based resolver.
Let $w \in L(\aut, C)$ and let $\rho = \hstratiterprime(w)$ be the sequence of transitions induced by $\hstratprime$ for $w$.
An induction over the prefixes~$\rho'$ of $\rho$ shows that $(q_{\rho'},\vec{v}_{\rho'}, q_{\rho'},\vec{v}_{\rho'})$ is in $W$, where $q_{\rho'}$ is the last state of $\rho'$ and $\vec{v}_{\rho'}$ is the extended Parikh image of $\rho'$, relying on the following two facts:
\begin{itemize}
    \item $\hstratprime$ only returns transitions that \myquot{stay} in the winning region.
    \item In a safety game, \attacker (the antagonist aiming to reach an unsafe state) cannot \myquot{leave} the winning region with their moves.
\end{itemize}

Thus, $\hstratprime$ will only pick residual transitions, as these are the ones that \myquot{stay} in the winning region. 
Now, an argument as the one in the proof of Lemma~\ref{lemma_onetokencorrectness} shows that $\hstratprime$ is indeed a resolver, as it only returns residual transitions.
\end{proof}

We conclude this section with the obvious question: Is the winning region of the safety game always semilinear? 
A classical result characterizes the winning region of the antagonist in a safety game (the player aiming to reach an unsafe state) 
as the attractor of the unsafe states. 
In our case, the attractor is the infinite union~$\bigcup_{i\geqslant0} \attr_i$ where $\attr_0$ is the set of unsafe states and $\attr_{i+1} = \cpre(\attr_i)$.
Here, $\cpre(\attr_i)$ denotes the controlled predecessor of the set~$\attr_i$ of states containing
\begin{itemize}
    \item all vertices of the antagonist having a successor in $\attr_i$, and
    \item all vertices of the protagonist having only successors in $\attr_i$.
\end{itemize}
Thus, if one can show that the attractor (the winning region of the antagonist) is semilinear, then, due to closure of semilinear sets under complementation, also the winning region of the protagonist is semilinear.
Then, Theorem~\ref{thm_resolversviagames} implies the existence of a \dpa-implementable resolver. 

Now, $\attr_0$, the set of unsafe states in the safety game, is semilinear, as semilinear sets are closed under Boolean operations and the set of safe states is defined as a Boolean combination of semilinear sets.
Further, one can show that in the safety games we consider here (but not in general), the controlled predecessor of a semilinear set is semilinear as well.
Thus, the only remaining obstacle is the infinite union.
However, semilinear sets are in general \emph{not} closed under infinite unions (as every set can be described as an infinite union of singletons, which are semilinear).
Nevertheless, the game graph the safety game is played in is rather simple, i.e., the edge relation depends only on a finite abstraction of the vertices (i.e., it only depends on the letter stored in a vertex) and is therefore monotone.\footnote{Note that related questions about winning regions in exact slow $k$-Nim are also open~\cite{knim}.}

Note that if the inherent structure of the safety games we consider here suffices to show that the winning region is indeed semilinear, then the \dpa's implementing a resolver cannot be effectively constructed, as it is undecidable whether a \pa is history-deterministic (Theorem~\ref{thm_hdnessundec}).

\section{Conclusion}
\label{sec_conc}
In this work, we have introduced and studied history-deterministic Parikh automata. 
We have shown that their expressiveness is strictly between that of deterministic and nondeterministic \pa, incomparable to that of unambiguous \pa, but equivalent to history-deterministic \nrbcm{1}. 
Furthermore, we showed that they have almost the same closure properties as \dpa (complementation being the notable difference), and enjoy some of the desirable algorithmic properties of \dpa.

An interesting direction for further research concerns the complexity of resolving nondeterminism in history-deterministic Parikh automata.
We have shown that if the winning region of \defender in the one-token game is semilinear, then there is a \dpa-implementable resolver. 
In further research, we investigate whether the winning region is necessarily semilinear.
Note that the analogous result for history-deterministic pushdown automata fails: not every history-deterministic pushdown automaton has a pushdown resolver~\cite[Theorem~7.3]{LZ22}\cite[Lemma~10.2]{GJLZ24}.

Good-for-gameness is another notion of restricted nondeterminism that is very tightly related to history-determinism. In fact, both terms were used interchangeably until very recently, when it was shown that they do not always coincide~\cite{BL21}. Formally, an automaton~$\aut$ is good-for-games if every two-player zero-sum game with winning condition~$L(\aut)$ has the same winner as the game where the player who wins if the outcome is in $L(\aut)$ additionally has to construct a witnessing run of $\aut$ during the play.
This definition comes in two forms, depending on whether one considers only finitely branching (weak compositionality) or all games (compositionality).

Recently, the difference between being history-deterministic and both types of compositionality has been studied in detail for pushdown automata~\cite{GJLZ22}.
These results are very general and can easily be transferred to \pa and \nrbcm{1}. 
They show that for \pa, being history-deterministic, compositionality, and weak compositionality all coincide, while for \nrbcm{1}, being history-deterministic and compositionality coincide, but not weak compositionality. 

The reason for this difference can be traced back to the fact that \nrbcm{1} may contain transitions that do not move the reading head (which are essentially $\epsilon$-transitions), but that have side-effects beyond state changes, i.e., the counters are updated.
This means that an unbounded number of configurations can be reached by processing a single letter, which implies that the game composed of an arena and a \nrbcm{1} may have infinite branching.
So, while \hdpa and \nhdrbcm{1} are expressively equivalent, they, perhaps surprisingly, behave differently when it comes to compositionality.

\bibliographystyle{plain}
\bibliography{biblio}

\end{document}